\documentclass[12pt, a4paper,reqno]{amsart}
\usepackage[margin=1in]{geometry}
\usepackage{amssymb}
\usepackage{enumerate}
\usepackage{hyperref}
\usepackage{doi,url}
\usepackage{dsfont}
\usepackage{mathabx}

\usepackage[dvipsnames]{xcolor}

\definecolor{purple}{rgb}{.5,0,1}
\definecolor{orange}{rgb}{1,.5,0}
\definecolor{pink}{rgb}{1,0,.5}

\usepackage{fancyvrb}
\usepackage{amsmath,amsthm}
\usepackage{enumitem}
\usepackage{thmtools}
\usepackage[style=apalike,style=numeric,maxbibnames=99,firstinits=true,citestyle=numeric,url=true,doi=true,
backend=bibtex
]
{biblatex}
\renewbibmacro{in:}{}
\AtEveryBibitem{\clearfield{number}}
\DeclareFieldFormat*{volume}{\mkbibbold{#1}}
\DeclareFieldFormat[article]{pages}{#1}

\numberwithin{equation}{section}
\newtheorem{theorem}{Theorem}[section]

\newtheorem{lemma}[theorem]{Lemma}
\newtheorem{corollary}[theorem]{Corollary}

\newtheorem{definition}[theorem]{Definition}
\newtheorem{remark}[theorem]{Remark}

\newtheorem*{theorem*}{Theorem}
\newtheorem*{lemma*}{Lemma}
\newtheorem*{remark*}{Remark}

\DeclareMathOperator{\supp}{supp}
\DeclareMathOperator{\tr}{tr}
\DeclareMathOperator{\Ran}{Ran}

\DeclareMathOperator{\Rank}{Rank}

\DeclareMathOperator{\dist}{dist}

\DeclareMathOperator{\diam}{diam}

\newcommand\R{\mathbb R}
\newcommand\N{\mathbb N}
\newcommand\C{\mathbb C}
\newcommand\Z{\mathbb Z}
\newcommand\K{\mathcal{K}}

\newcommand\cW{\mathcal{W}}

\newcommand\cF{\mathcal{F}}

\newcommand\e{\mathrm{e}}

\renewcommand\P{\mathbb P}
\newcommand\E{\mathbb E}
\newcommand\cE{\mathcal{E}}

\newcommand\cN{\mathcal{N}}

\newcommand\cG{\mathcal{G}}

\newcommand\cJ{\mathcal{J}}
\newcommand\cX{\mathcal{X}}

\newcommand\cH{\mathcal{H}}

\newcommand\cA{\mathcal{A}}

\newcommand\vphi{\varphi}

\newcommand{\pr}{\prime}

\newcommand\what{\widehat}
\newcommand\wtilde{\widetilde}

\newcommand\emp{\emptyset}

\newcommand\beq{\begin{equation}}
\newcommand\eeq{\end{equation}}
\newcommand\be{\begin{equation}\begin{aligned}}
\newcommand\ee{\end{aligned}\end{equation}}

\newcommand{\abs}[1]{\left\lvert #1 \right\rvert}
\newcommand{\norm}[1]{\left\lVert #1 \right\rVert}

\newcommand{\set}[1]{\left\{ #1 \right\}}
\newcommand{\pa}[1]{\left( #1 \right)}

\newcommand{\fl}[1]{\left\lfloor #1 \right\rfloor}

\newcommand{\cl}[1]{\lceil #1 \rceil}
\renewcommand{\fl}[1]{\left\lfloor #1\right \rfloor}

\newcommand\La{\Lambda}

\newcommand{\eq}[1]{\eqref{#1}}

\newcommand{\up}[1]{^{\left(#1\right)}}

\newcommand{\qtx}[1]{\quad\text{#1}\quad}
\newcommand{\mqtx}[1]{\; \ \text{#1}\; \  }
\newcommand{\sqtx}[1]{\;\text{#1}\;}

\newcommand{\bD}{\boldsymbol{\Delta}}

\newcommand{\tfd}{\pa{1- \tfrac{1}{\Delta}}}
\newcommand{\fd}{1- \frac{1}{\Delta}}
\newcommand{\sfd}{1- \tfrac{1}{\Delta}}
\newcommand{\nfd}{\pa{1- \frac{1}{\Delta}}}

\newcommand{\prr}{{\pr\pr}}

\begin{document}

\title[Localization in the random XXZ  spin chain]{Localization in the random XXZ  quantum spin chain}

\author{Alexander Elgart}
\address[A. Elgart]{Department of Mathematics; Virginia Tech; Blacksburg, VA, 24061-1026, USA}
 \email{aelgart@vt.edu}

\author{Abel Klein}
\address[A. Klein]{University of California, Irvine;
Department of Mathematics;
Irvine, CA 92697-3875,  USA}
 \email{aklein@uci.edu}

\thanks{A.E. was  supported in part by the NSF under grant DMS-1907435 and the Simons Fellowship in Mathematics Grant 522404.}


\begin{abstract} 
We study the many-body localization (MBL) properties of the Heisenberg XXZ spin-$\frac12$ chain in a random magnetic field. We prove that the system  exhibits  localization in any  given energy interval at the bottom of the spectrum in a nontrivial  region of the parameter space.
This region, which  includes weak interaction and strong disorder regimes,  is independent of the size of the  system and depends  only on the energy  interval. Our approach is based on the reformulation of the localization problem as an expression of  quasi-locality for functions of the random many-body XXZ Hamiltonian. This allows us to extend the fractional moment method for proving localization, previously derived in a single-particle localization context, to the many-body setting.
\end{abstract}


\setcounter{tocdepth}{1}
\maketitle
\tableofcontents	

\section{Introduction}

The last two decades have seen an explosion of physics research on the behavior of isolated quantum systems in which both disorder and interactions are present. The appearance of these two features has been linked to the existence of materials  that fail to thermalize and consequently cannot be described using equilibrium statistical mechanics. 
These materials are presumed to remain  insulators at non-zero temperature, a phenomenon called  many-body localization   (MBL). We refer the reader to the  physics reviews \cite{NandHuse,AL,abanin2019} for the general description of this phenomenon.  MBL-type behavior
  has been  observed  in cold atoms experiments \cite{Schreiber,Lukin}. 
 The stability of the MBL phase for infinite systems and all times remains a topic of intense debate \cite{sierant20,Kiefer,Sels,morningstar,sierant22}.

In this paper we    consider  the  random  spin-$\tfrac12$ Heisenberg XXZ chain in the Ising phase, a one-dimensional random quantum spin system.  This is the most studied model in the context of MBL both in the  physics and mathematics literature (going back to \cite{PPZ,HP}). It can be  mapped by the Jordan-Wigner transformation into an interacting spinless fermionic model closely related to the disordered Fermi-Hubbard Hamiltonian, a paradigmatic model in condensed matter physics that provides crucial insights into the electronic and magnetic properties of materials.  One  interesting feature of the  random  {\it one-dimensional} XXZ quantum spin system  is the emergence of  a many-body localization-delocalization transition. (In contrast,  prototypical non-interacting one-dimensional random Schr\"odinger operators do not exhibit a phase transition and are completely localized.) Numerical evidence for this transition in the disordered XXZ model has been provided in a number of simulations (e.g., \cite{Agarwal,Luitz,HP,bardarson,baygan}), but remains contested on theoretical grounds (e.g., \cite{deRoeck}).

 Until quite recently, mathematical results related to the proposed MBL characteristics, including zero-velocity Lieb-Robinson bounds, exponential clustering, quasi-locality, slow spreading of information,  and area laws, have been confined  to quasi-free systems. 
  The latter are models whose study can effectively be reduced to one of a (disordered) one-particle Hamiltonian. Examples of such systems include the XY spin chain in a random transversal field (going  back to \cite{KP}; see \cite{ARNSS} for a review on this topic), the disordered Tonks-Girardeau gas \cite{SeiringerWarzel}, and systems of quantum harmonic oscillators \cite{NSS}.  Another direction of research considers the effect of many-body interaction on a single-particle localization (rather than MBL) within the framework of the effective field theories.  This allows to consider a realistic Hilbert space for a single particle, such as $\ell^2(\Z^d)$, rather than finite dimensional ones that are typically used in the MBL context. In particular, the persistence of the dynamical localization in the Hartree-Fock approximation  for the disordered Hubbard model has been established in \cite{Ducatez,MaSch}. 

In the last few years, there has been some (modest) progress in understanding genuine many-body systems, all of which is concerned with the XXZ model, either in the quasi-periodic setting (where the exponential clustering property for the ground state of the Andr\'e-Aubry model has been established \cite{Mas1,Mas2}), or  in the droplet spectrum regime in the random case  \cite{BeW,EKS1}. In the latter case several  MBL manifestations have been established, including some that have never been previously  discussed in the physics literature  \cite{EKS3}.

While not exactly  solvable,  the XXZ spin chain  does have a symmetry, namely it preserves the particle number.  This enables a reduction to an infinite system of discrete $N$-body Schr\"odinger operators on the fermionic subspaces of $\Z^N$ \cite{NSt,FiSt}.   For the  XXZ spin chain in the Ising phase,
in the absence of a magnetic field the low energy eigenstates above the ground state are characterized by a {\it droplet regime}. In this regime spins form a droplet, i.e., a single cluster of down spins (particles) in a sea of up spins.
This reduction has been  effectively exploited   inside the droplet spectrum (the interval $ I_1$ in \eq{Ikle} below) using methods that resemble  the fractional moment method  for random  Schr\"odinger operators, yielding the small number of rigorous results   \cite{BeW,EKS1}.  However, these methods seem to be inadequate above this energy interval (i.e, inside the multi-cluster spectrum), and a new set of ideas that do not rely on a reduction to Schr\"odinger operators are required to tackle this case.  

In this paper we extend the energy interval for which MBL holds  well beyond the droplet spectrum, deep inside the multi-cluster spectrum. We develop a suitable method,   formulated and proved in terms of spin systems concepts.  In particular, our  method does  not rely on the  
reduction of the XXZ Hamiltonian to a direct sum of Schrodinger operators (and the subsequent analysis that uses single-particle tools).

Localization phenomenon in  condensed matter physics is usually associated with non-spreading of wave packets in a disordered medium. Experimentally, it is observed in semiconductors whose properties are predominantly caused by crystal defects or impurities, as well as in the variety of other systems. This phenomenon is by now well understood for quantum single  particle  models.  A prototypical system studied in this context is the Anderson Hamiltonian $H_A$, which is a self-adjoint operator acting on the Hilbert space $\mathcal H=\ell^2(\Z^d)$  of the form $H_A=-\Delta+\lambda V_\omega$. Here $\Delta$ is the (discrete) Laplacian describing the kinetic hopping, $V_\omega$ is a randomly generated multiplication operator ($\omega$ is the random parameter) describing the electric potential, and $\lambda$ is a parameter measuring the strength of the disorder.

Let us denote by $\delta_x\in\mathcal H$ the indicator of $x\in\Z^d$, and fix the random parameter $\omega$.  An  important feature of  $H_A$ as a map on $\ell^2(\Z^d)$  is its locality, meaning  $\langle \delta_x,H_A\delta_y\rangle=0$  if $\abs{x-y}>1$.
As a consequence, the resolvent  $(H_A-z)^{-1} $ retain a measure of locality, which we will call quasi-locality,  given by  the Combes-Thomas estimate
\be
\abs{\langle \delta_x,(H_A-z)^{-1} \delta_y\rangle}\le C_z \e^{-m_z\abs{x-y}},
\ee
where $C_z$ and $m_z$ are constants independent of $\omega$  such that  $C_z <\infty$ and $m_z >0$ if $z\in \C$ is outside the spectrum of $H_A$.  Maps given by smooth  functions  of $H_A$  also express a measure of quasi-locality, namely 
\be\label{eq:locman1}
\abs{\langle \delta_x,f(H_A)\delta_y\rangle} \le C_{f,n} \pa{1 + \abs{x-y}}^{-n},
\ee
where  $C_{f,n} <\infty$ for all $n\in\N$ and infinitely differentiable functions $f$.   Moreover, these quasi-locality estimate hold with the same constants for the restriction $H_A^\La$ of $H_A$ to a finite volume $\La \subset \Z^d$.  (See, e.g.,  \cite{Kirsch,Remling,GKdecay}.)

The two mainstream approaches for proving  localization in the single particle setting, namely the multi-scale analysis  (MSA) and the fractional moment method (FMM), going back to \cite{FS,FMSS,DK} and \cite{AM,A}, respectively, establish localization  for the (random) Anderson model $H_A$  by proving    quasi-locality estimates for  the finite volume resolvent  inside the spectrum of $H_A$.  In particular,  the fractional moment method shows that, fixing $s\in (0,1)$, for 
 large disorder $\lambda$ we have 
  \be
\E\set{  \abs{\left\langle {\delta_x,(H^\La_A-E)^{-1} \delta_y } \right\rangle}^s}\le C\e^{-m \abs{x-y}}, 
\ee 
for all  finite $\La \subset \Z^d$, $ x,y\in \La$, and energies $E\in \R$, where the constants $C<\infty$ and $m>0$ are in dependent of $\La$.   Moreover, one also gets a  quasi-locality   estimate for Borel functions of $H^A$  (dynamical localization),
\be
\E\set{ \sup_{f} \abs{\left\langle {\delta_x,f(H^\La_A)\delta_y } \right\rangle}}\le C\e^{-m \abs{x-y}}, 
\ee 
where the supremum is taken over all Borel functions on $\R$ bounded by one. 
Various manifestations of one-particle localization, such as non-spreading of wave packets; vanishing of conductivity in response to electric field;  statistics of the spacing between nearby energy levels, can be derived from these quasi-locality estimates.  (See, e.g.,  \cite{AW}.)  On the mathematical level, the quasi-locality   estimates provides an effective description of  single particle localization.

 The  MSA and the FMM prove localization for random Schr\"odinger operators, both in the  discrete and continuum settings. We refer the reader to the lecture notes  \cite{Kirsch,Kparis} and the monograph \cite{AW} for an introduction to the multi-scale analysis and the fractional moment method, respectively.

Both  methods  have been extended to  quantum system consisting of an arbitrary, but fixed,  number of interacting  particles, showing that 
many characteristics  of single-particle localization  
remain valid in this case  (e.g., \cite{CS,AW2,KlN1}).	
But  truly many-body systems (where the number of particles is proportional to the system's size) present new challenges. 
A major difficulty lies in the fact that the concepts of MBL proposed in the physics literature are not easily tractable on the mathematical level and it is not clear what could be chosen as the fundamental description of the theory from which other properties can be derived, as in a single particle case. For example, the available concept of quasi-locality in the many-body systems  looks very different from the one for single particle quantum systems.

To introduce a simple many-body system Hamiltonian, we consider a finite graph ${\bf \Gamma}=\pa{\mathcal V,\mathcal E}$ (where $\mathcal V$ is  the set of vertices  and $\mathcal E$ is the set of edges) and a family $\set{\mathcal H_i}_{i\in\mathcal V}$  of  Hilbert spaces. The Hilbert space of the subsystem associated with a  set $X\subset \mathcal V$ is given by $\mathcal H_X=\bigotimes_{i\in X}\mathcal H_i$, and the full Hilbert space
(we ignore particles' statistics)  is $\mathcal H_{\mathcal V}$. For each $X\subset \mathcal V$ one introduces the {\it algebra of observables} $\mathcal A_X$ measurable in this subsystem, which is the collection $\mathcal B(\mathcal H_X)$ of bounded linear operators on the Hilbert space $\cH_X$.  An observable $\mathcal O\in \mathcal A_{\mathcal V}$ is said to be supported by $X\subset\mathcal V$ if $\mathcal O=\mathcal O_{X}
\otimes \mathds{1}_{\mathcal H_{\mathcal V\setminus X}}$, where $\mathcal O_{X}\in\mathcal A_{X}$, i.e., if $\mathcal O$ acts trivially on $\mathcal H_{\mathcal V\setminus X}$.  Slightly abusing the  notation, we will usually identify $\mathcal O$ with $\mathcal O_{X}$, and call $X$ a support for $\mathcal O$.  
Since we are primarily interested here in understanding the way particles interact, the structure of a single particle Hilbert space $\mathcal H_i$ will be only of marginal importance for us. So we will be considering the simplest possible realization of such system, where each $\mathcal H_i$ is the two dimensional vector space $\C^2$ describing a  spin-$\frac12$ particle. 

We next describe the interactions between our spins. We again are going to consider the simplest possible arrangement, where only nearest neighboring spins are allowed to interact. Explicitly, for each pair of vertices $(i,j)\in\mathcal V$ that share an edge (i.e., $\set{i,j} \in \cE$), we pick an observable (called an {\it interaction}) $h_{i,j}\in \mathcal A_{\set{i,j}}$ such that  $h_{i,j}=h^*_{i,j}$,  an observable (called a {\it local transverse field}) $v_i=v_i^*\in \cA_{\set{i}}$,  and  associate a Hamiltonian $H^{\mathcal V}=\sum_{\set{i,j} \in\mathcal E}h_{i,j}+\sum_{i\in {\mathcal V}} v_i$ with our spin system. In particular, $H^{\mathcal V}$ is the sum of local observables and is consequently referred to as a {\it local Hamiltonian}.  Locality is manifested  by  $[[H^{\mathcal V},\mathcal O],\mathcal O^\prime]=0$ for any pair of observables  $\mathcal O\in \mathcal A_X$,  $\mathcal O^\prime\in \mathcal A_Y$,  with  $\dist(X,Y)>1$.  (To compare it with the concept of (single particle) locality for the map $H_A$, we need to define a local observable for the space $\ell^2(\Z^d)$. We will say that an observable $\mathcal O\in \mathcal L(\ell^2(\Z^d))$ has  support $X\subset\Z^d$ if $\mathcal O=\mathcal O_{X}
\oplus 0_{\ell^2(\Z^d\setminus X)}$ with $\mathcal O_{X}\in\mathcal L(\ell^2(X))$.
With this definition  locality of the map $H_A$, i.e., the property $\langle \delta_x,H_A\delta_y\rangle=0$  whenever $\abs{x-y}>1$ is equivalent to the statement that $[[H_{A},\mathcal O],\mathcal O']=0$ for any pair of observables  $\mathcal O,\mathcal O'$ with $\dist(\supp(\mathcal O),\supp(\mathcal O'))>1$.)

 The  XXZ spin chain is defined as above  on finite subgraphs $\Lambda$ of the graph $\Z$ (see Section \ref{secmodel}).  
Consider  $\La\subset \Z$ connected, and let $\abs{\La}$ be its cardinality.   We say we have a particle at  the site   $i\in \La$ if we have spin down in the copy $\cH_i$ of $\C^2$.   
 Let  $\cN_i$ be the orthogonal projection onto configurations with a particle at the site $i$, and set 
$[i]^\La_p= \set{j\in \La, \abs{j-i} \le p}$ for $p=0,1,\ldots$.   Given $B\subset \La$, let $P_+^B$ be the orthogonal projection onto configurations with no particles in $B$. 
 In the Ising phase $H^\Lambda$ is  a $2$-local, gapped, frustration-free system,  and $P_+^\La$ describes the projection onto the ground state of $H^\La$ (see Remark~\ref{remIsing}).

We can now informally state our main results.  
We first prove that the resolvent $R^\La_z=(H^\La-z)^{-1}$ 
exhibits quasi-locality in the form (see Lemma~\ref{lem:locdet'} and Remark~\ref{reamrkHquasi})
\be\label{eq:ql}
\norm{\cN_i R_z^ \La P_+^{[i]^\La_p}}\le C_z \e^{-m_z p},
\ee
where $C_z$ and $m_z$ are constants,  independent of $\La$ and of the  transverse field,  such that  $C_z <\infty$ and $m_z >0$ if $z\in \C$ is outside the spectrum of $H^\La$.
We also   establish  the many-body analogue of \eqref{eq:locman1}: 
\be\label{eq:locmany1}
\norm{\cN_i f(H^ \La) P_+^{[i]^\La_p}}\le C_{f,n} \pa{1 +p}^{-n},
\ee
where  $C_{f,n} <\infty$ for all $n\in\N$ and infinitely differentiable functions $f$ on $\R$ with compact support.
(See  Appendix \ref{app:quasil}.)

 We next consider the random  XXZ spin chain (see Definition~\ref{defXXZ}). The relations \eqref{eq:ql}-\eqref{eq:locmany1} suggest, by  analogy with random Schr\"odinger operators,  that localization should be manifested as quasi-locality inside the spectrum of  $H^\La$. This is indeed what we prove in  Theorem~\ref{thm:maintech}.  We  introduce  increasing energy intervals  $I_{\le k} $, $k=0,1,2,\ldots$, in \eq{Ikle}   and  prove that   quasi-locality of the form given in \eq{eq:ql} holds for the resolvent for energies in  $I_{\le k} $  for any
 fixed $ k$. In particular, given $s\in(0,\frac 1 3)$ we prove, in the appropriate ($k$ dependent)  parameter region, that
\be\label{eq:quasran2} 
\E\set{\norm{\cN_i R_E^ \La P_+^{[i]^\La_p}}^s}\le C_k\abs{\Lambda}^{\xi_k} \e^{-m_k p} \qtx{for all}  E\in I_{\le k} ,
\ee
where the constants $C_k<\infty$, $ \xi_k >0$, $m_k>0$ do not depend on $\La$. As a consequence 
we  derive a quasi-locality estimate for Borel functions of $H^\La$ (Corollary \ref{thm:eigencor}):
 \be\label{eq:quasran} 
\E\pa{\sup_f {\norm{\cN_i f(H^\La)P_+^{[i]^\La_p}}}}\le C_k\abs{\Lambda}^{\xi_k} \e^{-m_k p},
\ee  
where the supremum is taken over all Borel functions Borel functions on $\R$ that are equal to  zero outside the interval $I_{\le k}$ and bounded by one.

While the  estimates \eqref{eq:quasran2} and \eqref{eq:quasran}  are very natural from the mathematical perspective, it is far from obvious  whether  they  yield  any of the  MBL-type features proposed by physicists.   
Nevertheless,   in a sequel to this paper \cite{EK23}, we derive slow propagation of information, a putative MBL manifestation,   from  Theorem~\ref{thm:maintech}  and  Corollary \ref{thm:eigencor}, for any $k\in \N$.

In the  droplet spectrum,   \cite[Theorem 2.1]{EKS1} imply Corollary \ref{thm:eigencor} (with $k=1$), and a converse can established using \cite[Remark~3.3]{EK23}.  While \cite{EKS1} and the follow-up paper \cite{EKS3} contain several  MBL-type properties such as the (dynamical) exponential clustering property,   (properly defined) zero-velocity Lieb-Robinson bounds, and  slow propagation (non-spreading) of information, they are all derived using \cite[Theorem 2.1]{EKS1}  as the starting point. We stress that  \cite[Theorem 2.1]{EKS1}, by its very nature, can only hold in the droplet regime, so, while it provides us with very strong consequences    in the $k=1$ case, we do not  expect  the methods of \cite{EKS1,EKS3}  to be of any use in the multi-cluster case, that is, for   $k\ge 2$.

  Although the methods derived in this work are not universal (which is typical for many-body results), they are sufficiently powerful for investigation of MBL phenomena in this context,  as shown in \cite{EK23}. We have to admit however that in the physics literature MBL is usually associated with energies that are not fixed (as we assumed in this work) but are comparable with the system size $\abs{\Lambda}$. We do not expect that our techniques will be sufficient to probe such energies. To be able to do so would require non-perturbative techniques similar to the ones use in the  investigations of one dimensional random Schr\"odinger operators.

The model description and main results (Theorem \ref{thm:maintech} and Corollary \ref{thm:eigencor})
 are presented  in Section \ref{secmodel}. { In Section \ref{sec:ir'} we outline the main ideas used in the proof of Theorem \ref{thm:maintech}, which is completed in Section \ref{sec:prmain}.  Corollary \ref{thm:eigencor} is proven in Section \ref{sec:cor}. 
  Appendix~\ref{appA} contains some useful identities.   Appendix \ref{app:quasil} contains the proof of of the many-body quasi-locality estimate \eq{eq:locmany1}.

 Throughout the paper, we will use generic constants $C, c,m$, etc., whose values will be allowed to change from line to line, even  in a displayed equation. These constants will not depend on subsets of $\Z$, but    they will,  in general depend on the parameters of the model   introduced in Section \ref{secmodel}  (such as $\mu$, $k$, $\Delta_0$, $ \lambda_0$, and  $s$). When necessary, we will indicate the dependence of a constant on  $k$ explicitly by writing it as $C_k, m_k$, etc. These constants can always be estimated from the arguments, but we will not track the changes to avoid complicating the arguments.

\section{Model description and main results}\label{secmodel}

\subsection{Model description}\label{secmodel}

The random XXZ quantum spin-$\frac 12$  chain  on an finite subset $\Lambda$ of $\Z$ is given by a self-adjoint  Hamiltonian $H^\Lambda$ acting on the finite dimensional  Hilbert space  $\mathcal H_\Lambda=\otimes_{i\in \Lambda} \cH_i$, where  $\cH_i=\C^2$ for each $i\in\Lambda$. For a vector $\phi \in \C^2$ we let   $\phi_i$ denote the vector as an element of $\cH_i$;  for an operator  ($2\times 2$ matrix) $A$ on $\C^2$ we let $A_i$ denote the operator acting on $\cH_i$.

We consider only finite subsets of $\Z$, so by a subset of $\Z$ we will always mean a finite subset.   If  $S\subset T\subset \Z$, and $A_S$ is an operator on  $\cH_S$, we consider $A_S$ as operator on $\cH_T$ by identifying it with  $A_S \otimes \mathds{1}_{T\setminus S}$, where $ \mathds{1}_R$ denotes the identity operator on $\cH_R$.  We thus identify   $\cA_S$ with a subset of $\cA_T$,  where $\cA_R$ denotes the algebra of bounded operators on $\cH_R$.

We now fix $\La \subset \Z$,  and  consider $\La$ as a subgraph of $\Z$. We 
 denote by $\dist_\La$ the  graph distance in $\La$, which  can be infinite if $\La$ is not a connected subset of $\Z$. We write $K^c=\La\setminus K$ for $K\subset \La$.
To define $H^\Lambda$ we  introduce some notation and definitons.
\begin{enumerate}
\item By $\sigma^{x,y,z}$ and $\sigma^\pm=\frac 12( \sigma^x \pm i \sigma^y)$ we will denote the standard Pauli matrices and  ladder operators, respectively.

\item  By $\uparrow\rangle: = \begin{pmatrix} 1  \\ 0  \end{pmatrix} $ and  $\downarrow\rangle = \begin{pmatrix} 0  \\ 1 \end{pmatrix} $ we will denote  the elements of the canonical basis of $\C^2$, called   spin-up  and spin-down, respectively. Letting 
$\mathcal{N} = \tfrac{1}{2} (\mathds{1}-\sigma^z)$, we note that   $\mathcal{N} \uparrow\rangle=0$ and   $\mathcal{N} \downarrow\rangle=\downarrow\rangle$, and interpret $\downarrow\rangle$ as a particle.

\item $\mathcal{N}_i$, the matrix $ \cN$ acting on $\cH_i$,  is   the projection onto the spin-down state (also called the local number operator) at site $i$. Given $S\subset \Lambda$, 
 $\cN_{S} = \sum_{i\in S} \mathcal{N}_i$ is   the total (spin-down) number  operator in $S$.

\item   The total number operator $\cN_\La$  has eigenvalues $0,1,2,\ldots, \abs{\La}$.  ($\abs{S}$  denotes the cardinality of $S\subset \Z$.) We set  $\cH_\Lambda\up{N}=\Ran\pa{\chi_N(\mathcal N_\Lambda)}$,  obtaining   the Hilbert space decomposition 
$ \cH_\La= \bigoplus_{N=0}^{\abs{\La}} \cH_\La\up{N}$.  We will use the notation $\chi^\La_{N}=\chi_{\set{N}}(\mathcal N_\Lambda)$. 

 \item The canonical (orthonormal) basis $\Phi_\La $  for $\cH_\La$ is constructed as follows:  Let $\Omega_\Lambda=\phi_\emptyset=\otimes_{i \in \Lambda}\uparrow\rangle_i$ be the vacuum state. Then
 \be \label{eq:standbasis}
& \Phi_\La 
=\set{\phi_A=\pa{\prod_{i\in A}\sigma_i^-}\Omega_\Lambda:\ A\subset \Lambda} =\bigcup_{N=0}^{\abs{\La}}  \Phi_\La\up{N}, 
 \ee
 where $\Phi_\La\up{N}= \set{\phi_A:  A\subset \Lambda, \ \abs{A}=N}$.  Note that $ \Phi_\Lambda\up{0}=\set{\Omega_\Lambda}$.  
\end{enumerate}

We now define the free XXZ quantum spin-$\frac 12$  Hamiltonian  on   $\Lambda\subset \Z$ by
\beq\label{eq:modH}
H_0^\Lambda=H_0^\Lambda(\Delta)=-\tfrac{1}{2\Delta} \bD^\Lambda + \cW^\Lambda \qtx{on} \cH_\La,
\eeq
where 
\begin{align}\label{bD}
 \bD^\Lambda &= \sum_{\set{i,i+1}\subset \Lambda} \pa{\sigma_i^+\sigma_{i+1}^-+\sigma_i^-\sigma_{i+1}^+},\\
 \cW^\Lambda & = \cN_\Lambda -\sum_{\set{i,i+1}\subset \Lambda} \cN_i\cN_{i+1} ,
 \end{align} 
and  $\Delta>1$ is the anisotropy   parameter, specifying the Ising phase ($\Delta=1$ selects the Heisenberg chain and  $\Delta=\infty$ corresponds to the  the Ising chain).

 We will consider the XXZ model in the presence of a transversal field $\lambda V^\Lambda_\omega$, given by $V^\Lambda_\omega=\sum_{i\in \Lambda} \omega_i \mathcal{N}_i$, where $\omega_i\ge 0$, and the parameter $\lambda>0$ is used to modulate the strength of the field. The full  Hamiltonian is then 
\beq\label{eq:H}
H^\Lambda=H^\Lambda_\omega =H^\La_\omega(\Delta,\lambda)= H_0^\Lambda(\Delta)+\lambda V_\omega^\Lambda= -\tfrac{1}{2\Delta} \bD^\Lambda + \cW^\Lambda+\lambda V_\omega^\Lambda.
\eeq

 \begin{remark}\label{remark1}
 \begin{enumerate}

\item The operator  $\bD^\La$ can be viewed  as the analog of the Laplacian operator on $\cH_\La$. 

\item  $\cN_i$  is diagonalized by the canonical basis for all $i\in \La$:   $\cN_i \phi_A= \phi_A$ if $i\in A$ and $0$ otherwise. It follows that
 the total number operator $\cN_\La$ is also  diagonalized by the canonical basis: $\cN_\La \phi_A= \abs{A} \phi_A$. 
\item $\cW^\La$,  the number of clusters operator, is diagonalized by the canonical basis: $\cW^\La \phi_A= W^\La_A \phi_A$,
where  $W^\La_A\in [0,\abs{A}]\cap\Z$ is the number of clusters of $A$ in $\La$, i.e., the number of connected components of  $A$ in $\Lambda$ (considered as a subgraph of $\Z$). 

\item $V^\Lambda_\omega$ is   diagonalized by the canonical basis: $V^\Lambda_\omega  \phi_A= \omega\up{A} \phi_A$,  where $\omega\up{A}=\sum_{i\in A} \omega_i$.

\item  The operators $\cN_\La$, $\cW^\La$, and $V^\Lambda_\omega$ commute.

\item The XXZ Hamiltonian  $H^\Lambda$  preserves the total particle number,  
\beq
[H^\Lambda,\cN_\Lambda]= -\tfrac{1}{2\Delta}[ \bD^\Lambda,\cN_\Lambda]  =0. 
\eeq

\end{enumerate}
\end{remark}

We will consider the XXZ model  in the presence of a random transversal field, that is, $\omega = \set{\omega_i}_{i\in\Z}$  is a family of  random variables. More precisely, we make the following definition.

\begin{definition}\label{defXXZ} The random XXZ  spin Hamiltonian on $\La \subset \Z$ is the operator
$H^\La=H^\La_\omega(\Delta,\lambda)$  given in \eq{eq:H}, where
$\Delta >1$, $\lambda >0$, and
$\omega = \set{\omega_i}_{i\in\Z}$  is a family of independent identically distributed random variables, whose  common probability  distribution $\mu$ satisfies
\beq\label{mu}
  \set{0,1}\subset \supp \mu\subset[0,1] ,
  \eeq 
  and  is assumed to be  absolutely continuous with a bounded density.
\end{definition}

From now on $H^\La$ always denotes the  random XXZ  spin Hamiltonian on $\La$. 
The corresponding resolvent is given by $R_E^\Lambda=\pa{H^\Lambda-E}^{-1}$, which is  well defined for  almost every energy $E\in\R$.    We set $\omega_S = \set{\omega_i}_{i\in S}$ for $S\subset \Z$, and denote the corresponding expectation and probability by $\E_S$ and $\P_S$. 

It is convenient to introduce the local interaction terms
\beq\label{tildeh}
{h}_{i,i+1}= - \mathcal{N}_{i}\mathcal{N}_{i+1} -\tfrac{1}{2\Delta}\pa{\sigma_i^+\sigma_{i+1}^-+\sigma_i^-\sigma_{i+1}^+},
\eeq
which allows us to rewrite 
\beq\label{eq:modH1}
H_0^\Lambda=\sum_{\set{i,i+1}\subset  \Lambda} {h}_{i,i+1} + \cN_\Lambda.
\eeq

It can be verified that on $\cH_{\set{i,i+1}}=\cH_i^2\otimes  \cH_{i+1}^2$ we have 
\begin{align}\label{Nsigma}
{\tfrac 12 \pa{\cN_i+\cN_{i+1}}-\cN_i\cN_{i+1}}
 \mp \tfrac 12 \pa{\sigma_i^+\sigma_{i+1}^-+\sigma_i^-\sigma_{i+1}^+} \ge 0,
 \end{align}
which implies that   $\cW_\La \pm \tfrac 1 2\bD_\La  \ge 0$, that is,
\beq\label{cWbD}
- 2\cW_\La  \le - \bD_\La \le 2\cW_\La  .
\eeq
It follows that 
\beq\label{H0W}
\tfd \cW^\Lambda\le H_{0}^\Lambda \le\pa{1+\tfrac 1 \Delta} \cW^\Lambda, \qtx{so} \tfd \cW^\Lambda\le H^\Lambda.
\eeq
We conclude that the spectrum of $H^\Lambda$ is  of the form  
\beq\label{eq:spH}
\sigma(H^\Lambda)=\set{0} \cup \pa{\left[1 -\tfrac 1 \Delta, \infty \right ) \cap  \sigma(H^\Lambda) }.
\eeq

The lower bound in \eq{H0W}  suggests the introduction of the energy thresholds $k\tfd$, $ k=0, 1,2\ldots$.  We define the energy intervals
 \be\label{Ikle}
\what I_{\le k}&=\left(-\infty, (k+1)\tfd\right), \quad 
\what I_{ k}= \left[\sfd, (k+1)\tfd\right),\\
I_{\le k}&= \left(-\infty, (k+\tfrac 3 4)\tfd\right) ,\quad 
I_k= \left[\sfd, (k+\tfrac 3 4)\tfd\right).
\ee 
We call  $\what I_{ k}$  the $k$-cluster spectrum.

Given $\emptyset \ne S\subset \Lambda\ $, we  define the orthogonal projections  $P_\pm^{S}$ on $\cH_\Lambda$ by 
 \begin{align}\label{P+-S}
P_+^{S}=\bigotimes_{i\in S}\pa{\mathds{1}_{\mathcal H_{i}}-\cN_i}= \chi_{\set{0}}\pa{\cN_{S}}  \qtx{and} 
P_-^{S}={   \mathds{1}_{\cH_\La} }-P_+^{S}=  {  \chi_{\N}}\pa{\cN_{S}}.
\end{align}
$P_+^{S}$  is the orthogonal projection onto states with no particles in the set $S$;  $P_-^{S}$  is the orthogonal projection onto states with at least one particle in $S$.
We also set 
 \be \label{P+-empty}
P_+^{\emptyset } ={   \mathds{1}_{\cH_\La} }\qtx{and } P_-^{\emptyset }=0.
 \ee
   \begin{remark}\label{remIsing} In the Ising phase, i.e., $\Delta>1$,  we have \eq{H0W} and  \eqref{eq:spH} for all $\La \subset \Z$.  It follows that  the XXZ chain Hamiltonian $H^\La$  has ground state $\Omega_\La$ and the ground state energy is $0$ ($H_\La  \Omega_\La$=0), and, moreover, the ground state energy  is gapped.  This makes  $H^\Lambda$  a $2$-local, gapped, frustration-free system.  These features, plus the preservation of the total particle number,  make the XXZ model especially amenable to analysis. In particular, the number of eigenstates of $H^\Lambda$ in  the intervals $I_{\le k}$    grows only polynomially in the volume of $\La$ (not exponentially as the dimension of $\cH_\La$)  as shown in  Lemma \ref{lem:bndHSQ'} below.
\end{remark}

}

\subsection{Main results}

Our main result establishes  quasi-locality for the resolvent  of the random XXZ chain  inside the spectrum of  $H^\La$.

\begin{theorem}[Quasi-locality for resolvents]\label{thm:maintech} 
 Fix  $\Delta_0>1$, $ \lambda_0 >0$, and let $s\in(0,\frac 1 3)$.  Then for all   $k\in\N^0$  there exist  constants $D_k,F_k,\xi_k, \theta_k>0$ (depending on $k$, $\Delta_0$, $ \lambda_0$ and  $s$) such that,
 for all $\Delta \ge \Delta_0$ and $\lambda \ge \lambda_0$ with  $\lambda \Delta^2\ge D_k$,  $\Lambda\subset  \Z$ finite, and  energy $E\in I_{\le k}$, we have
\be\label{eq:mainbnd}
\E\set{\norm{P_-^{A}R_E^\Lambda P_+^{B}}^s}\le F_k\abs{\Lambda}^{\xi_k} \e^{-\theta_k\dist_\La(A,B^c)} ,
\ee 
for   $A\subset B\subset \La$ with    $A$ connected in $\La$.
\end{theorem}

The  theorem is proven in Section \ref{sec:prmain}.

\begin{remark} If $A$ is not connected in $\La$, the theorem still holds with \eq{eq:mainbnd} replaced by
\be\label{eq:mainbndnot}
\E\set{\norm{P_-^{A}R_E^\Lambda P_+^{B}}^s}\le F_k  {\Upsilon}^\La_A \abs{\Lambda}^{\xi_k} \e^{-\theta_k\dist_\La(A,B^c)}  ,
\ee 
where ${\Upsilon}^\La_A$ denotes the number of connected components of $A$ in $\La$. This follows from  \eq{eq:mainbnd}  and 
\beq
P_-^{A}=\sum_{j=1}^{{\Upsilon}^\La_A} P_+^{\bigcup_{i=i}^{j-1} A_i}P_-^{A_j},
\eeq
where
  $A_j$, $j=1,2,\ldots, {\Upsilon}^\La_A$,  are  the  connected components of $A$ in $\La$ .

\end{remark}

As a consequence of Theorem \ref{thm:maintech}  we prove the following quasi-locality estimate for Borel functions of $H^\La$. By $B(I_{\le k})$ we denote the collection of Borel functions on $\R$ that are equal to  zero outside the interval $I_{\le k}$.

\begin{corollary}[Quasi-locality for Borel functions]\label{thm:eigencor} Assume the hypotheses and conclusions of
 Theorem \ref{thm:maintech},  Then for all   $k\in\N^0$  there exist  constants $\wtilde F_k,\wtilde \xi_k, \wtilde \theta_k>0$ (depending on $k$, $\Delta_0$, $ \lambda_0$ and  $s$) such that,
 for all $\Delta \ge \Delta_0$ and $\lambda \ge \lambda_0$ with  $\lambda \Delta^2\ge D_k$, {  and $\Lambda\subset  \Z$ finite,}  we have
 \be\label{eq:eigencor} 
\E_{\Lambda}\pa{\sup_{\substack{f\in B(I_{\le k}):\\\|f\|_\infty\le1}}\norm{P_-^{A}f(H^\Lambda) P_+^{B}}}\le \wtilde F_k\abs{\Lambda}^{\wtilde \xi_k}  \e^{-\wtilde \theta_k\dist_\La(A,B^c)} ,
\ee  
for all $A\subset B\subset \La$, $A$ connected in $\La$.
\end{corollary}

The proof of the Corollary  is given in  Section \ref{sec:cor}.

\section{Key ingredients for the proofs}\label{sec:ir'}

In this section we collect a number of definitions, statements and   lemmas that will facilitate the proof of Theorem \ref{thm:maintech}.  

$\La$ will always denote a finite subset of $\Z$ and $A\subset \La$ will always denote a nonempty subset  connected  in $\La$. ($B\subset \La$, $S\subset \La$, etc.,  may not be connected in $\La$.)

\subsection{Some definitions}

   \begin{itemize}

  \item  Given  $ M\subset \Lambda$  and $q\in\Z$, we define  enlarged (for $q\ge 0$) and  trimmed  (for $q<0$)  set $[M]^\La_q $ by
   \be
{[M]^\La_q } &:=\begin{cases}\set{x\in\La: \dist_\La\pa{x,M}\le q} &\sqtx{if}  q\in \N^0=\set{0} \cup \N  \\ \set{x\in\La:\dist_\La\pa{x,  M^c}\ge 1-q}= M\setminus [M^c]^\La_{-q}&\sqtx{if}  q\in-\N\\
\set{x\in \La:  \dist_\La\pa{x,M}<\infty}= \bigcup_{p\in \N^0} [M]^\La_p & \sqtx{if}q=\infty
\end{cases}.
\ee
Note that $ [M]^\La_{-\abs{M}}=\emptyset$.  Moreover,  $[M]^\La_\infty= [M]^\La_{\abs{\La}-1} $ is the connected component of $\La$ containing $M$,  and   we have 
 \be\label{HM0}
[H^\La, P_\pm^{[M]^\La_\infty}]=0.
\ee

We   define $\partial_{ex}^\La M$  (the external boundary of $M$ in $\La$) ,   $\partial_{in}^\La M$ (the inner boundary of $M$ in $\La$), and    $ \partial^\La M$ (the boundary of $M $ in $\La$), by 

\be
\partial_{ex}^\La M&:=\set{x\in\Lambda:\ \dist_\La\pa{x,M}=1}= [M]^\La_1\setminus M,\\
 \partial_{in}^\La M&:=\set{x\in\Lambda:\ \dist_\La \pa{x, M^c}=1}=M\setminus [M]^\La_{-1},\\
 \partial^\La M &:=  \partial_{in}^\La M \cup  \partial_{ex}^\La M.
\ee

 It follows that
\be
]M[^\La_q:= [M]^\La_{q+1}\setminus  [M]^\La_{q}&=\begin{cases}  {\partial^\La_{ex}{[M]^\La_{q}}},  & q\in \N^0\\
{\partial^\La_{in }{[M]^\La_{q+1}}} & q\in-\N
 \end{cases},
\ee 
and we have
\be\label{MpM1p}
]M[^\La_p= ] M^c[^\La_{-p-1} \qtx{for} p\in\Z.
\ee

If $M=\set{j}$ we write ${[j]^\La_q }={[\set{j}]^\La_q }$.

\item   Given $A\subset B \subset \La$,  we let $\rho^\La (A,B)$  be the largest $q\in \N^0 \cup \set{\infty}$ such that $ [A]^\La_q \subset B$, that is,
 \be\label{defrhoLa}
 \rho^\La (A,B)= \sup \set{q\in \N^0:  [A]^\La_q \subset B} =\dist_\La (A, B^c)-1.
 \ee
 It will be more convenient to use $ \rho^\La (A,B)$  instead of $\dist_\La (A, B^c)$ in the proofs.

Note that 
\be\label{rhoinfty}
\rho^\La (A,B)=\infty \iff \dist_\La (A, B^c)=\infty \iff  [A]^\La_\infty \subset B.
\ee

\item  It follows from \eq{HM0}  and \eq{rhoinfty} that
\be\label{PAPB=0}
P_-^A R_{E}^{\Lambda} P_+^{B}=0 \qtx{if} A\subset B\subset \La \qtx{and}\rho^\La(A,B)=\infty,
\ee
so it suffices to prove Theorem \ref{thm:maintech} for $\rho^\La(A,B)<\infty$.  Moreover,
since $A \subset B$ we have $[A]^\La_{\rho^\La(A,B)}\subset B$, and hence
\be
\norm{P_-^{A}R_E^\Lambda P_+^{B}}\le \norm{P_-^{A}R_E^\Lambda P_+^{[A]^\La_{\rho^\La(A,B)}}},
\ee
so without loss of generality it suffices to prove  \eq{eq:mainbnd} for  $B=[A]^\La_{\rho}$ with   $\rho \in \N^0$.

\item  Given  $K\subset \La$, we    consider the operator  $H^K= H^K\otimes \mathds{1}_{\cH_{K^c}}$ acting on $\cH_\La$. We also
 consider the operators on $\cH_\La$ given by
\be\label{eq:Gamma}
H^{K,K^c}= H^{K}+H^{K^c}, \quad  R^{K,K^c}_E=\pa{H^{K,K^c}-E}^{-1}, \quad \Gamma^K=H^\La-H^{K,K^c}.
 \ee

\end{itemize}

\subsection{Quasi-locality for resolvents}
 
  The following lemma and remark yields (deterministic) quasi-locality for the resolvent  of the   XXZ chain  outside the spectrum of  $H^\La$.

\begin{lemma} \label{lem:locdet'} 
Let ${\Theta}\subset {\Lambda}$, and consider the Hilbert space $\cH_{\Lambda}$.   Let  the operator $T\in\mathcal A_{\Lambda}$  be  of the form 
\be\label{eq:T}T=T^{\Theta}+T^{{ \Theta^c}};\qtx{where}T^{\Theta}\in\mathcal A_{\Theta}\qtx{and} T^{{ \Theta^c}}\in\mathcal A_{{ \Theta^c}},
\ee
 and let $\cX\in  \cA_{\Lambda}$ be  a projection such  that  $[ \cX,T]=0$ and $[\cX, P_\pm^K]=0 $ for all $K\subset {\Theta}$.

Suppose
\begin{enumerate}

\item  For all   $K \subset {\Theta} $   we have   $[P_-^{K},T]P_+^{[K]_1^{\Theta}}=0$.

 \item For all     $K \subset {\Theta} $, with $K$ connected in $\Theta$,   we have  $\norm{[P_-^{K},T]}\le  \gamma  $. 
 
 \item  $T_\cX$, the restriction of the operator $T$ to $\Ran \cX$, is invertible with
 $ \norm{T_\cX^{-1}}_{\Ran \cX} \le \eta^{-1}$,  where $\eta>0$.

 \end{enumerate}

Then for all $A\subset B\subset {\Theta}$, with  $A$  connected in ${\Theta}$, 
we have
 \be \label{eq:localitybndet}
\norm{P_-^{A}\,T_\cX^{-1}\,P_+^{B}}_{\Ran \cX}\le \eta^{-1}  \e^{-m \rho^{\Theta} (A,B)}   , \qtx{with} 
m=\ln\pa{{\gamma^{-1}}\eta }.
\ee 
\end{lemma}

\begin{proof}

 We consider first the case $\cX=\mathds{1}_{\cH_\La}$.
Let $A\subset B\subset {\Theta}$, with  $A$  connected in ${\Theta}$. 
 Let  $1\le t \le \rho^{\Theta} (A,B)$, so    $[A]_{t}^\Theta \subset B$. We have
\be 
P_-^{  A}\,T^{-1}\,P_+^{ B}=T^{-1}[T,P_-^{  A}]T^{-1}\,P_+^{B}={T^{-1}[T,P_-^{ A}]} P_-^{ [A]^\Theta_1}T^{-1}\,P_+^{ B}, 
\ee
 using condition (i) of  the Lemma. Proceeding recursively, we get
\be
P_-^{ A}\,T^{-1}\,P_+^{B}= \pa{\prod_{p=0}^{t-1}{T^{-1}[T,P_-^{ [A]_p^\Theta}]}} P_-^{ [A]^\Theta_{t}}T^{-1}\,P_+^{B}.
\ee
 Since  $A$ is  connected in $\Theta$,   $[ A]_{r}^{\Theta}$, $r=1,2,\ldots,t$, are also connected in $\Theta$. Using assumptions (ii) and (iii), we get 
\be \label{PATt}
\norm{P_-^{A}\,T^{-1}\,P_+^{ B}}\le   \pa{ \gamma    \eta^{-1}}^t \eta^{-1}.
\ee
Since \eq{PATt} holds for all $1\le t \le \rho^{\Theta} (A,B)$, we get
{   \be \label{eq:localitybndetI}
\norm{P_-^{A}\,T^{-1}\,P_+^{B}} \le \eta^{-1}  \e^{-m \rho^{\Theta} (A,B)}   , \qtx{with} 
m=\ln\pa{{\gamma^{-1}}\eta }.
\ee }

If condition (iii) holds with a projection  $\cX\in  \cA_{\Lambda}$  such  that  $[ \cX,T]=0$ and $[\cX, P_\pm^K]=0 $ for all $K\subset {\Theta}$, 
 then
$\wtilde T= T\cX + \eta (1-\cX)$ satisfies conditions (i), (ii),  and condition (iii)  with  $\cX=\mathds{1}_{\cH_\La}$, and  the estimate  \eq{eq:localitybndetI} for $\wtilde T$ implies  \eq{eq:localitybndet}. 
\end{proof}

\begin{remark} \label{reamrkHquasi}  Lemma \ref{lem:locdet'} yields quasi-locality for the resolvent of the operator $H^\La$.
The operator  $H^\La-z$ satisfies the hypotheses of Lemma~\ref{lem:locdet'}   for $z\notin \sigma(H^\La)$, with 
$\Theta=\La$,  $\gamma = \frac 1 {\Delta}$  (use   \eq{PGamma1}),  $\cX=\mathds{1}_{\cH_\La}$,  and  $\eta= \dist (z,\sigma(H^\La))$.  It follows that,  with $R_z^\La= (H^\La -z)^{-1}$, for all $A\subset B\subset \La$,
 we have
 \be \label{eq:localitybndetfree}
\norm{P_-^{A}R_z^ \La P_+^{B}}\le \pa{\dist (z,\sigma(H^\La)) }^{-1}  \e^{-m \rho^{\Theta} (A,B)}   , \sqtx{with} 
m=\ln\pa{\Delta \dist (z,\sigma(H^\La))}.
\ee 
\end{remark}

From now on we fix  $\Delta_0>5$, $ \lambda_0 >0$, and assume $\Delta \ge \Delta_0$ and $\lambda \ge \lambda_0$. The constants will depend on $\Delta_0$ and $\lambda_0$.

  Given  $m \in \N^0$, we set $Q_m^\Lambda=\chi_{\set{m}}\pa{\cW^\Lambda}$, the orthogonal projection onto  configurations  with exactly $m$ clusters, and let  $Q_B^\Lambda=\chi_{B}\pa{\cW^\Lambda}=\sum_{m\in B} Q_m^\Lambda$  for $ B\subset \N^0$. Note that  $Q_0^\Lambda=  P_+^\Lambda$ and  $Q_\N^\Lambda= \chi_{\N}(\cN^\Lambda)$.
 For $k\in \N$, we set 
\be\label{QkhatQ}
Q_{\le k}^\Lambda   =Q_{\set{1,2,\ldots,k}}^\Lambda =\sum_{ m=1}^k Q_m^\Lambda \qtx{and} 
\what Q_{\le k}^\Lambda   =Q_{\le k}^\Lambda + \tfrac {k+1} k Q_0^\Lambda.
\ee
We also set  
\be \label{eq:compH'}
\what  H_0^{ \Lambda}&=H^{\Lambda}+\tfd Q_0^\Lambda,\\
\what  H_k^{ \Lambda}&=H^{\Lambda}+{k}\tfd \what  Q_{\le k}^{\Lambda} \qtx{for}  k\in \N.
\ee 
We use the notation
\be
\what  R^{\Lambda}_{k,E}&= \pa{\what  H_k^\La  -E}^{-1}  \mqtx{for} E \notin \sigma(\what  H_k^\La), \ k\in \N^0.
\ee

 It follows from \eqref{H0W} and \eq{Ikle} that   for $k\in \N^0$ we have
 \be \label{eq:hatH1'}
\what  H_k^{ \Lambda}\ge  \pa{k+1} \tfd    \qtx{and}  \pa{\what  H_k^{ \Lambda}-E}  \ge \tfrac 1{4}\tfd  \mqtx{for} E \in I_{\le k}.
\ee

For   $k\in \N^0$  and $E\in I_{\le k}$,
 the operator $T= \what  H_k^{ \Lambda}-E$   satisfies the assumptions of Lemma \ref{lem:locdet'} with $\Theta=\La$, 
$\gamma = \frac 1 {\Delta}$,    $\cX=\mathds{1}_{\cH_\La}$, and 
$\eta= \frac 1{4}\nfd $ (see  \eq{eq:hatH1'}). In this case
 $  m =\ln   \frac  {\Delta-1}{4}$, and  hence for  $A\subset B\subset {\La}$, 
  \eq{eq:localitybndet} yields 
\be  \label{decest}
\norm{P_-^{A}\what  R^{\Lambda}_{k,E}P_+^{B}}\le \tfrac  4  {\fd}\e^{-\pa{\ln  \frac {\Delta-1} {4}}\rho^{\La} (A,B)}.
\ee

To have decay in \eq{decest}, we need $ \frac {\Delta-1} {4}>1$, that is, $\Delta >5 $. In the proof of Theorem~\ref{thm:maintech}, we will fix $\Delta_0 >  5 $ and $\lambda_0>0$, and require $\Delta \ge \Delta_0$ and $\lambda \ge \lambda_0$. In this case, we have $\tfrac 4  {\fd} \le  \tfrac 4  {1-\frac 1 {\Delta_0}} $ and $ \ln  \tfrac {\Delta-1} 4\ge  \ln  \tfrac {\Delta_0-1} 4$, so we have
\be  \label{CT}
\norm{P_-^{A}\what  R^{\Lambda}_{k,E}P_+^{B}}\le C_0  \e^{-m_0 \rho^\La (A,B)},
 \sqtx{with} C_0=  \tfrac 4  {1-\frac 1 {\Delta_0}}, \; m_0= \ln  \tfrac {\Delta_0-1} 4>0.
\ee

  It follows from \eq{HM0}, which also holds for the operator $\what  H_k^{ \Lambda}$,  that  
\be\label{PMPtildeM}
P_-^M R_{E}^{\Lambda} P_+^{[M]^\La_\infty}=0 \qtx{and}P_-^M\what  R_{{k,E}}^{\Lambda} P_+^{[M]^\La_\infty}=0 \qtx{for} M\subset \La .
\ee

\begin{remark} \label{remarkDelta5} We will prove Theorem~\ref{thm:maintech} with $\Delta_0> 5$ to simplify our analysis. The proof  can be extended to  arbitrary $\Delta_0> 1$ with  minor modifications.  Specifically,  for $1<\Delta_0 \le  5$ we need to improve the decay rate in \eq{decest}, which is derived from the lower bound in \eq{eq:hatH1'}. To do so, we would  replace $\what H_k^{ \Lambda}$ in the proof by  $\what H_{k+r}^{ \Lambda}$, where $r\in \N$, so  \eq{eq:hatH1'} yields  $\what H_{k+r}^{ \Lambda}-E \ge (r+ \frac 1 4)\tfd$ for $E\in  I_{\le k}$, leading to $m_0= \ln \pa{(r + \frac 1 4) \pa{\Delta_0-1}}>0 $ for an appropriate choice of $r$.
\end{remark}

\subsection{An a-priori estimate}
The first step toward the proof of Theorem \ref{thm:maintech} is to understand why the expression on the left hand side of \eqref{eq:mainbnd} is actually finite. A useful technical device  for this purpose is the following bound, where
$\norm{T}_{HS}$ denotes the Hilbert-Schmidt norm of the operator $T$.

\begin{lemma}[A-priori estimate]\label{lem:apri}     Let $i,j\in\Lambda$  ($i=j$ is allowed) and let $T_{1},T_2$ be a pair of  Hilbert-Schmidt operators on $\cH_\La$ that are $\omega_{\set{i,j}}$-independent. Then we have 
\beq\label{eq:weak1-1} 
\E_{\set{i,j}}\pa{\norm{T_1\cN_iR_E^\Lambda\cN_j T_2}_{HS}^s}  \le
C\lambda^{-s}\norm{T_1}^s_{HS}\norm{T_2}^s_{HS} \sqtx{for all} E\in \R \sqtx{and} s\in (0,1) .
\eeq
\end{lemma}

The lemma  follows from \cite[Proposition 3.2]{AENSS}, used with $U_1=\cN_j$, $U_2=\cN_k$ there,   and the layer-cake representation for a non-negative random variable $X_\omega$:
$\E(X_\omega^s)=\int_0^\infty \P(X_\omega>t^{1/s})\, dt$ for $s\in (0,1)$. 

The  Hilbert-Schmidt operators for Lemma~\ref{lem:apri} are provided by the following result.

 \begin{lemma}\label{lem:bndHSQ'}
 Let $k\in \N$. Then
\begin{align}\label{trXk}
\norm{Q_{\le k}^{\Lambda}}_{HS}&\le \sqrt{k} \abs{\Lambda}^{k},\\
\label{trkH}
\tr \chi_{\what I_{\le k}}(H^\La)&\le  k\abs{\Lambda}^{2k}+1.
\end{align}
\end{lemma}

\begin{proof}

For $m\ge 1$ and  $N \ge 1$ we have the rough estimate
\be 
\tr \chi^\La_N Q_{m}^{\Lambda}\le \abs{\Lambda}^m N^{m-1}.
\ee
Thus
\begin{align} 
\tr \chi^\La_N Q_{\le k}^{\Lambda}\le \sum_{m=1}^k \abs{\Lambda}^m N^{m-1}= \tfrac 1 N \tfrac {(\abs{\Lambda} N)^{k+1}-(\abs{\Lambda} N)}{(\abs{\Lambda} N)-1}\le  k \abs{\Lambda}^k N^{k-1}.
\end{align}
It follows that
\begin{align} 
\tr Q_{\le k}^{\Lambda}\le k \abs{\Lambda}^k \sum_{N=1}^{\abs{\Lambda}}    N^{k-1}\le k\abs{\Lambda}^{2k}. 
\end{align}

To prove \eq{trkH}, 
let $\what{H}^\La_k$ be as in \eq{eq:compH'},  and note  that \eq{eq:hatH1'}  implies 
$\tr \chi_{\what I_{\le k}}(\what{H}^\La_k)=0$.   Since the spectral shift is bounded by the rank of the perturbation,     it follows from \eq{eq:compH'}  that
\beq
\tr \chi_{\what I_{\le k}}(H^\La)\le  \tr \chi_{\what I_{\le k}}(\what{H}^\La_k)+\Rank \pa{ k \tfd \what Q_{\le k}^{\Lambda}}= \tr  \what Q_{\le k}^{\Lambda}= \tr  Q_{\le k}^{\Lambda}+1.
\eeq  
\end{proof}

Lemmas~\ref{lem:apri} and  \ref{lem:bndHSQ'} yield the a priori estimate
  \be\label{eq:apr}
\E_{\set{i,j}}{\norm{Q^\Lambda_{\le k}\cN_i  R^{\Lambda}_{E}\cN_j Q^\Lambda_{\le k}}_{HS}^s}\le C\lambda^{-s}k ^s\abs{\Lambda}^{2sk}  \sqtx{for all} i,j\in \La \sqtx{and} s\in (0,1) .
\ee 
More generally, we have
  \be\label{eq:aprAB}
\E_{\set{A\cup B}}{\norm{Q^\Lambda_{\le k}P_-^A R^{\Lambda}_{E}P_-^B Q^\Lambda_{\le k}}_{HS}^s}\le C\lambda^{-s}k ^s\abs{\Lambda}^{2sk}\abs{A}\abs{B}\qtx{for} \emptyset \ne A,B\subset\La.
\ee

Those  a priori estimates are only useful if we can "dress" the resolvent with factors of $Q^\Lambda_{\le k}$ on both sides. To be able to do so, we will decorate $R^{\Lambda}_{E}$ with  resolvents  of positive operators that satisfy the quasi-locality property.

\subsection{Dressing  resolvents with Hilbert-Schmidt operators} 
For $k=1,2,\ldots$, and $E\in I_{\le k}$,   we  use the resolvent identity 
 \be\label{eq:resmodl}
R_{E}^{\Lambda}=\what  R^{\Lambda}_{k,E}+ k\tfd R^{\Lambda}_{E} \what Q^\Lambda_{\le k}\what  R^{\Lambda}_{k,E}=\what  R^{\Lambda}_{k,E}+ k\tfd\what  R^{\Lambda}_{k,E}\what Q^\Lambda_{\le k}  R^{\Lambda}_{E}.
\ee
Using it twice we get 
 \be\label{eq:resmod'1}
R_{E}^{\Lambda}=\what  R^{\Lambda}_{k,E}+ k\tfd\what  R^{\Lambda}_{k,E} \what Q^\Lambda_{\le k} \what  R^{\Lambda}_{k,E}+k^2\tfd^2\what  R^{\Lambda}_{k,E}\what Q^\Lambda_{\le k}   R^{\Lambda}_{E} \what Q^\Lambda_{\le k}\what  R^{\Lambda}_{k,E}.
\ee

 We use the notation $(p)_+=\max\pa{p,0}$ for $p\in \R$.

\begin{lemma} \label{lem:inred}
 
 Let $\cX$ denote a spectral projection of $\cN_\La$ (say,  $\cX=\mathds{1}_{\mathcal H_{\La}}$ or $\cX= \chi_N^\La$).  Let $A\subset B\subset \La$, 
  and  $1\le t=\rho^\La\pa{A, B}<\infty$.   Let  $E\in I_{\le k}$ and let $m_0$ be as in \eq{CT}.
  
  \begin{enumerate}
\item We have the following estimate on operator norms:
 \be\label{Fpq}
&{\norm{\cX P_-^AR_E^\Lambda P_+^B}}\le  C_k\Big(\abs{\Lambda}\e^{- m_0 t}
 +    \sum_{p=-\abs{A}}^{\abs{\Lambda}} \sum_{q=-\abs{A}}^{\abs{\Lambda}}\e^{-m_0(p)_+}\e^{-m_0\pa{t-q-1}_+}\norm{\cX F^\La_{p,q}(E,A)}\Big),\\
&\text{where}\quad  F^\La_{p,q}(E,A)=Q^\Lambda_{\le k}P_+^{[A]^\La_{p}}P_-^{]A[^\La_{p}}  R^{\Lambda}_{E} P_+^{[A]^\La_q}P_-^{]A[^\La_{q}}Q^\Lambda_{\le k} \qtx{for} p,q\in \Z.
\ee 

\item We have the following estimates on Hilbert-Schmidt norms: 
\be\label{eq:inredaHS}
\norm{\cX P_-^AR_E^\Lambda P_+^B Q^\Lambda_{\le k}}_{HS}   \le 
C_k\Big(\abs{\Lambda}^{k}\e^{-m_0 t}+  \sum_{q=-\abs{A}}^{\abs{\Lambda}}  \e^{-m_0\pa{q}_+} {\norm{\cX  Q^\Lambda_{\le k}P_-^{]A[_{q}}R^{\Lambda}_{E}P_+^BQ^\Lambda_{\le k}}}_{HS}\Big).
\ee
Moreover, for $s\in (0,1)$ we have 
\be\label{eq:inredaHSE}
\E\pa{\norm{\cX P_-^AR_E^\Lambda P_+^B Q^\Lambda_{\le k}}_{HS} ^s} \le  C_{k,s} \abs{\Lambda}^{2sk+3}.
\ee
\end{enumerate}
\end{lemma}

\begin{proof}

 Let $A\subset B\subset \La$, $A$ connected in $\La$.    Since $\cX$ commutes with all the relevants operators, we will just do the proof for $\cX=I$.

Using \eqref{eq:resmod'1}, \eq{QkhatQ}, and \eqref{CT} we get
\be\label{eq:mainbndmo1}
{\norm{P_-^AR_E^\Lambda P_+^B}}\le C_0\e^{-m_0t}&+k{\norm{P_-^A\what  R^{\Lambda}_{k,E}Q^\Lambda_{\le k}\what  R^{\Lambda}_{k,E} P_+^B}}+ k^2 {\norm{P_-^A\what  R^{\Lambda}_{k,E} Q^\Lambda_{\le k}  R^{\Lambda}_{E}Q^\Lambda_{\le k}\what  R^{\Lambda}_{k,E}P_+^B}}.
\ee

  Using \eq{PMPtildeM} , \eqref{eq:stidena}, and  the fact that $Q^\Lambda_{\le k}$ commutes with  $P_\pm$ operators,  we get
\be\label{secter56}
&{\norm{P_-^A\what  R^{\Lambda}_{k,E} Q^\Lambda_{\le k+1}\what  R^{\Lambda}_{k,E} P_+^B}}\le \sum_{q=-\abs{A}}^{\abs{\Lambda}}  {\norm{D_q}}{\norm{E_q}},
\ee
where
\be\label{eq:mainbndmo3'}
D_q=P_-^A\what  R^{\Lambda}_{k,E} P_+^{[A]_{q}} \mbox{ and }  E_q=P_-^{]A[_{q}}\what  R^{\Lambda}_{k,E}P_+^B.
\ee
Using  \eq{eq:hatH1'}, \eqref{CT},  and   $ ]A[_{q} \subset B$ for $q+1\le t$, we get
 \be\label{eq:mainbndmo8833}
  {\norm{D_q}}\le C_0 \e^{-m_0(q)_+} \sqtx{and} {\norm{E_q}}\le C_0 \e^{-m_0\pa{t-q-1}_+} \sqtx{for all} q\in \Z.
 \ee
It follows that
\be\label{secter88}
{\norm{P_-^A\what  R^{\Lambda}_{k,E}Q^\Lambda_{\le k+1}\what  R^{\Lambda}_{k,E} P_+^B}}\le  C_0^2 \sum_{q=-\abs{A}}^{\abs{\Lambda}}\e^{-m_0(q)_+}\e^{-m_0\pa{t-q-1}_+}\le   C_0^\pr \abs{\Lambda}\e^{- m_0 t}.
\ee

This leaves us with the estimation of the last term in \eqref{eq:mainbndmo1}. To this end, we use \eq{PMPtildeM},  \eqref{eq:stidena}, and  \eq{eq:mainbndmo8833} to obtain
\be\label{eq:mainbndmo88}
&{\norm{P_-^A\what  R^{\Lambda}_{k,E}Q^\Lambda_{\le k+1}  R^{\Lambda}_{E}Q^\Lambda_{\le k+1}\what  R^{\Lambda}_{k,E} P_+^B}}\le  \sum_{p=-\abs{A}}^{\abs{\Lambda}} \sum_{q=-\abs{A}}^{\abs{\Lambda}} {\norm{D_p}}\norm{F_{p,q}}{\norm{E_q}}\\
& \qquad  \qquad\le    C _0^2  \sum_{p=-\abs{A}}^{\abs{\Lambda}} \sum_{q=-\abs{A}}^{\abs{\Lambda}}\e^{-m_0(p)_+}\e^{-m_0\pa{t-q-1}_+}\norm{F_{p,q}},
\ee
where $F_{p,q}=F_{p,q}^\La (E,A)$ is as in \eq{Fpq} for $p,q \in \Z$.

Combining  \eq{eq:mainbndmo1}, \eq{secter88}, and  \eq{eq:mainbndmo88} we get \eq{Fpq}.

  To prove  \eqref{eq:inredaHS}, we proceed as in \eq{eq:mainbndmo1} using \eqref{eq:resmodl},  exploit $\norm{T_1T_2}_{HS}\le \norm{T_1}\norm{T_2}_{HS}$, and use \eq{trXk}, obtaining
\be
&\norm{P_-^AR_E^\Lambda P_+^B Q^\Lambda_{\le k}}_{HS} \le  C_k \e^{-m_0 t}\abs{\La}^k
+ k{\norm{P_-^A\what  R^{\Lambda}_{k,E}Q^\Lambda_{\le k} R^{\Lambda}_{E} P_+^B Q^\Lambda_{\le k}}}_{HS}.
\ee
We then use \eq{PMPtildeM},  \eqref{eq:stidena}, and \eq{eq:mainbndmo8833} to get
\be
{\norm{P_-^A\what  R^{\Lambda}_{k,E}Q^\Lambda_{\le k}  R^{\Lambda}_{E} P_+^B Q^\Lambda_{\le k}}}_{HS}
&\le  \sum_{q=-\abs{A}}^{\abs{\Lambda}} {\norm{D_q}} {\norm{Q^\Lambda_{\le k}P_-^{]A[_{q}}R^{\Lambda}_{E}P_+^BQ^\Lambda_{\le k}}}_{HS}
\\  &\le  \sum_{q=-\abs{A}}^{\abs{\Lambda}} C_0 \e^{-m_0\pa{q}_+} {\norm{Q^\Lambda_{\le k}P_-^{]A[_{q}}R^{\Lambda}_{E}P_+^BQ^\Lambda_{\le k}}}_{HS}.
\ee

Given  $s\in (0,1)$, it follows from  \eqref{eq:inredaHS} and \eq{eq:aprAB}  that

\be\label{eq:inredaHSEcalc}
&\E\pa{\norm{\cX P_-^AR_E^\Lambda P_+^B Q^\Lambda_{\le k}}_{HS}^s}  
\le C_{k,s} \abs{\Lambda}^{2sk+3}.
\ee
\end{proof}

 \subsection{Large deviation estimate}

 Using a large deviation argument we get the following refinement of \eqref{eq:aprAB}. Recall we may assume $\rho^\La(A,B)<\infty $  in view of \eq{PAPB=0}.

\begin{lemma}\label{lem:LD}
Let $k\in \N$.  Let $A\subset B\subset \La$, 
 with $\rho^\La(A,B)<\infty $.  Given  $s\in (0,\frac 12)$,
there exist constants $C_{k,s} ,c_\mu >0$ such that  for all $E \in I_{\le k}$ we have
\be\label{eq:aprLD} 
\E\pa{\norm{\chi^\La_NQ^\Lambda_{\le k}P_-^A R^{\Lambda}_{E}P_+^B Q^\Lambda_{\le k}}_{HS}^s} \le  C_{k,s}   \abs{\Lambda}^{2(sk+1)} \pa{\e^{-  c_\mu N}+   \e^{-m_0 \rho^\La  (A,B)}}.
\ee 
 In particular,
\be\label{eq:aprLD1} 
\E{\norm{\chi^\La_NQ^\Lambda_{\le k} P_-^A R^{\Lambda}_{E}P_+^B Q^\Lambda_{\le k}}^s} \le C_{k,s}\abs{\Lambda}^{2(sk+1)}  \e^{-m_{0,\mu}\rho^\La  (A,B)} \mqtx{if} 8kN \ge\rho^\La  (A,B),
\ee
where  $m_ {0,\mu}>0$.   
\end{lemma}

\begin{proof}

Recall $\cH_\La\up{N}= \Ran \chi^\La_N$, and let $\cH_\La\up{N,k}= \Ran \chi^\La_N Q^\Lambda_{\le k}$.
Recall also that  the restriction of  $V_\omega^\La$  to $\cH_\La\up{N}$  is  diagonalized by the canonical basis $ \Phi_\Lambda\up{N}$   as in Remark~\ref{remark1}(iii).

Let us first assume that $N$ is such that 
$N \lambda \bar{\mu}\ge  2k\tfd$, where $ \bar{\mu}$ denotes the mean of the probability distribution $\mu$ (see Definition~\ref{defXXZ}).
 The standard large deviation estimate (Cramer's Theorem)
gives
\be
\P\set{\lambda \omega\up{M} < k\tfd}\le \P\set{\omega\up{M}  < N \tfrac {\bar{\mu}} 2}\le \e^{- c_\mu N}\sqtx{for all} M\subset \La \sqtx{with} \abs{M}=N,
\ee
where $c_\mu$ is a constant depending only on the probability distribution $\mu$. This  implies that there exists $C_k>0$ such that
\be
\P\set{\lambda   \omega\up{M}< k\tfd}\le  C_k\e^{- c_\mu N}\sqtx{for all} N\in \N \sqtx{and } M\subset \La \sqtx{with} \abs{M}=N.
\ee

 It follows that for the event
\be
{\mathcal B_k^N}=\set{\exists  M\subset \La \sqtx{with} \abs{M}=N, \;  W^\La_M =k\sqtx{and} \lambda  \omega\up{M}< k\tfd},
\ee 
 we have  
 \be\label{PBkB'}
 &\P_{\Lambda}\pa{\mathcal B_k^N}   \le C_k\e^{- c_\mu N}\tr  Q_{\le k}^{\Lambda,N}  
 \le C_k \abs{\Lambda}^{2k} \e^{- c_\mu N}   \qtx{for} N=1,2\ldots, \abs{\Lambda}, 
  \ee
  where we also used Lemma \ref{lem:bndHSQ'}.
 On the complementary event  $\pa{\mathcal B_k^N}^c$  we have
  \beq\label{PVP'}
 \lambda  V_\omega   \chi^\La_N Q^\Lambda_{\le k}  \ge   k\tfd  \chi^\La_N Q^\Lambda_{\le k} .
  \eeq
  
  If \eq{PVP'} holds we conclude that 
  \be\label{PVP'HN}
H^{\La,N}&\ge \tfd \cW^\La   + \lambda   V_\omega=  \pa{Q_{\le k}^{\Lambda,N}  + Q_{\ge k+1  }^{\Lambda,N} }\pa{\tfd \cW^\La    +  \lambda  V_\omega}\\ & 
 \ge   \tfd  Q_{\ge k+1  }^{\Lambda,N} \cW^\La    +Q_{\le k}^{\Lambda,N}  \pa{\tfd \cW^\La   +  \lambda  V_\omega}\ge  (k+1) \tfd.
 \ee
 
We deduce that for $\omega \in \pa{\mathcal B_k^N}^c$   and  $E \in I_{\le k}$ we have
 \be
 H^{\La,N}-E\ge (k+1) \tfd- (k+\tfrac 34 ) \tfd = \tfrac 14  \tfd .
 \ee
 
 Proceeding as in  the derivation of  \eq{CT},
it follows from Lemma~\ref{lem:locdet'} 
and Remark~\ref{reamrkHquasi}, that for $\omega \in \pa{\mathcal B_k^N}^c$ we have, for  $A\subset B\subset {\La}$ with  $A$  connected in ${\La}$  that
 \be
{\norm{\chi^\La_NP_-^{A} R_E^\Lambda P_+^{B}}}\le C_0 \e^{-m_0 \rho^\La (A,B)}.
\ee

 Given $E \in I_{\le k}$, and letting $T =\chi_N Q_{\le k}^\Lambda P_-^{A} R_E^\Lambda P_+^{B}Q_{\le k}^\Lambda$,   we obtain
  \be
\E\pa{\norm{T}_{HS}^s} & \le \E\pa{\chi_{\mathcal B_k^N}\norm{T}_{HS}^s}+\E\pa{\chi_{\pa{\mathcal B_k^N}^c}\norm{T}_{HS}^s}\\ 
&\le \pa{\P \pa{\mathcal B_k^N}}^{\frac 1 2} \pa{\E\pa{\norm{T}_{HS}^{2s}}}^{\frac 12}+ 
C_0  \e^{-m_0 \rho^\La(A,B)}\norm{\chi_N Q_{\le k}^{\Lambda}}_{HS}^s
\\ 
&\le  C_{k,s}   \abs{\Lambda}^{2(sk+1)} \pa{\e^{- {\frac 1 2} c_\mu N}+   \e^{-m_0 \rho^\La (A,B)}},
\ee
where we used \eq{PBkB'},  Lemma \ref{lem:bndHSQ'}, and   \eq{eq:aprAB} with $2s$ instead of $s$. This estimate is  
\eq{eq:aprLD},  up to a redefinition of the constant $c_\mu$.

The estimate \eq{eq:aprLD1} follows immediately from \eq{eq:aprLD}. 
\end{proof}

\subsection{Decoupling of resolvents}\label{subsec:dec}
We now illustrate the basic idea that allows us to obtain the exponential decay of the left hand side in \eqref{eq:mainbnd}, analogous to the decoupling argument in the single particle localization literature. For this purpose, we will consider a more convenient object than the one in \eqref{eq:mainbnd}. To do so, let  $A\subset M\subset B\subset \La $,  and consider 
$P_+^{M^c} P_-^AR_E^\Lambda P_+^{B}.$
Let $K\subset \Z$ be such that   $M\subset [K]_{-1}\subset K\subset[K]_1\subset B$.
  The resolvent identity yields  (recall \eq{eq:Gamma}) 
 \be\label{eq:geompert}
 P_+^{M^c} P_-^{A}R_E^\Lambda P_+^{B}&
  = - 
 P_+^{M^c} P_-^{A}R^{K,K^c}_E \Gamma^K R_E^\Lambda  P_+^{B}
 =-   P_+^{M^c} P_-^{A}R^{K,K^c}_EP_+^{K^c} \Gamma^K R_E^\Lambda  P_+^{B}\\
 & =- P_+^{M^c} P_-^{A}R_E^{K}P_+^{K^c}\Gamma^K R^\Lambda_E P_+^{B},
 \ee
 where we used that   $P_-^{A} R^{K,K^c}_E P_+^{K}=0$ by \eq{HM0}  since      $[A]_\infty^K \subset K$, 
 $P_+^{M^c}R^{K,K^c}_E=P_+^{M^c} R^{K,K^c}_EP_+^{K^c}$ by \eq{HM0} since $K^c\subset M^c$, 
 and 
 $R^{K,K^c}_EP_+^{K^c}=R^{K}_EP_+^{K^c}$.
 Using  the specific structure of the XXZ Hamiltonian, that is,  \eqref{hPN}-\eqref{hPN98}, we have
 $P_+^{K^c}\Gamma^K= P_+^{K^c}P_-^{ \partial^\La K} \Gamma^K P_-^{ \partial^\La K}= P_+^{K^c}P_-^{\partial_{in}^\La K}\Gamma^K P_-^{\partial_{ex}^\La K} $, so it follows from \eq{eq:geompert} that
\be\label{eq:lonex2a}
P_+^{M^c} P_-^{A}R_E^\Lambda P_+^{B}=-P_+^{M^c} P_-^{A}R_E^{K}P_-^{\partial_{in}^\La K}P_+^{K^c}\Gamma^K P_-^{\partial_{ex}^\La K}R^\Lambda_E P_+^{B}.
\ee

We now use the resolvent identity
 for the operator $H^{ [K]_1^\La,([K]^\La_1)^c}$ 
 and  \eqref{hPN}, obtaining 
\be\label{eq:caa2}
  P_-^{\partial_{ex}^\La K}R_E^\Lambda P_+^{B}=-  P_-^{\partial_{ex}^\La K}R_E^\La  P_-^{\partial_{ex}^\La K} \Gamma^{[K]_1} P_-^{\partial_{ex}^\La [K]_1}  P_+^{[K]_1}R_E^{[K]_1^c} P_+^{B} .
\ee

Combining \eqref{eq:lonex2a}--\eqref{eq:caa2}, we obtain 
\be\label{eq:keydec}
&P_+^{M^c} P_-^{A}R_E^\Lambda P_+^{B}=
\\& (P_-^{A}{P_+^{M^c\cap K}R_E^{K}P_-^{\partial_{in}^\La K}})P_+^{K^c}\Gamma^K \pa{P_-^{\partial_{ex}^\La K}R_E^\La  P_-^{\partial_{ex}^\La K}}  \Gamma^{[K]_1} P_+^{[K]_1}{\pa{P_-^{\partial_{ex}^\La [K]_1} R_E^{[K]_1^c} P_+^{B\cap [K]_1^c}}}.
\ee
This is the basic decoupling formula, in a sense that the expressions in the first and last parentheses on the  last  line are statistically independent and of the same form as the left hand side of \eqref{eq:mainbnd}.  So, if we can perform the averaging over the random variables at sites $r\in \partial_{ex}^\La K$ to get rid of the middle resolvent, we will effectively decouple the system into pieces supported by the disjoint subsets  $K$ and $[K]_1^c$. (Note that these pieces do not depend on the random variables at sites $r\in \partial_{ex}^\La K$.) 
 This decoupling will be performed using the a-priori estimate \eqref{eq:aprAB}, after we dress the corresponding resolvents with  Hilbert-Schmidt operators on both sides as in Lemma \ref{lem:inred}. In broad strokes, we then will extract the (initial)  exponential decay from the expression in the first parenthesis in \eqref{eq:keydec} using reduction to lower energies and obtain the full exponential decay using a sub-harmonicity argument. We flesh out  details of this process  as we proceed with the proof.

\subsection{Clusters classification}\label{clusterclass}
In preparation to initiate the  FMM, we first inspect the structure of states in $\Ran Q^\Lambda_{\le k}$. Since $Q^\Lambda_{\le k}$ is a multiplication operator in the canonical basis $\set{\Phi_\Lambda\up{N}}_{N=0}^{\abs{\La}}$  introduced in \eqref{eq:standbasis}, we just need to consider the elements $\vphi_M$ of this basis with $M$ that belong to a set   $\mathcal S_{N,k}^\La:=\set{M\subset\Lambda:\ \abs{M}=N,\ 1\le  W^\La_M\le k}$, $N\ge 1$.   (Recall that $W^\La_M$ is the number of clusters of the configuration $M$, i.e., the number of connected components of  $M$ in the graph $\Lambda$.)   Denoting by  $ \pi_{\vphi}$  the orthogonal projection onto $\C\vphi$,   given  $M\in \mathcal S^{\La}_{N,k}$, we abuse the notation and write $\pi_M$ for $\pi_{\vphi_M}$, so
 $\pi_{M} = \pa{\prod_{j\in M}\cN_j} P_+^{M^c}$, and note that
$
 \chi_N^\La  Q^\Lambda_{\le k}=  \sum_{M\in \mathcal S_{N,k}^\La}  \pi_M.
$

Given $ A\subset \La$ ,   we set $\mathcal S^{\La,A}_{N,k}=\set{M\in \mathcal S^{\La}_{N,k}: \ M\cap A\ne \emptyset}$, and note that
$
 \chi_N^\La  Q^\Lambda_{\le k} P_-^A=  \sum_{M\in \mathcal S_{N,k}^{\La,A}}  \pi_M.
$
We set
\beq
\gamma_A(M) = \max_{x\in M}  \dist_\La(x,A) \le \diam_\La (M)= \max_{x,y\in M}  \dist_\La(x,y) \qtx{for} M\in \mathcal S^{\La,A}_{N,k}.
\eeq
Note that  $\diam_\La (M)=N-1$  for $k=1$ and   $\diam_\La (M)\ge N\ge 2$  for $k\ge 2$.

 If   $ 8kN < \rho^\La  (A,B)$ we will use the following lemma.

\begin{lemma}\label{lem:partM} Fix  $k\ge 2$. Let $A\subset B\subset\La$ be such that    $8kN < \rho^\La(A,B)<\infty$, and let  $M\in \mathcal S^{\La,A}_{N,k}$.

\begin{enumerate}
\item  Suppose $\gamma_A(M) < 4kN$.   Then setting  $Z=[A]^\La_{6kN}$, we have
\be \label{smallgamma}
A \cup M\subset [Z]_{-1}\subset Z \subset [Z]_1\subset B; \;\rho^\La(A\cup M,Z)\ge 2kN; \quad  \rho^\La(Z,B)\ge 2kN.
\ee

\item  Suppose $  \rho^\La(A,B)\le 2 \gamma_A(M)$.  Let $ d_\rho:= \fl{ \tfrac{ \rho^\La(A,B)}{6k}}$. Then there exists 
$a \in \set{1,2,\ldots, 3k-1}$,
 such that, letting  $K=  [A]^\La_{a d_\rho}$, we have
\be\label{eq:distMT}
   \rho^\La \pa{\partial^\La K, \La\setminus M} \ge  d_\rho -1 .
   \ee 
 Moreover, letting   $M_1= M\cap K$ and $M_2= M\cap K^c $, we have $K\subset B$ and   $M_{i}\neq\emptyset$ for $i=1,2$.

 \item  Suppose $ 8kN< 2\gamma_A(M)<  \rho^\La(A,B )$. Let $ d_\gamma:= \fl{ \frac{ \gamma_A(M)}{3k}}$. Then there exists $a \in \set{1,2,\ldots, 3k-1}$,
 such that, letting  
 \be\label{eq:K}
 K=  [A]^\La_{a d_\gamma}\cup\pa{[A]^\La_{\gamma_A(M)+d_\gamma}\setminus [A]^\La_{a d_\gamma+1}},
 \ee
 we have 
 \be\label{eq:distMT'}
  \rho^\La \pa{\partial^\La K, \La\setminus M}   \ge  d_\gamma -1 .
   \ee 
 Moreover, letting   $M_1= M\cap [A]^\La_{j d_\gamma}$ and $M_2= M\cap [A]^\La_{\gamma_A(M)}\setminus [A]^\La_{j d_\gamma+1}$, we have $M_1\cup M_2=M\subset K\subset B$ and  $M_{i}\neq\emptyset$ for $i=1,2$.  
 
\end{enumerate}
\end{lemma}

\begin{proof}

Part (i) is obvious. To prove  Parts (ii) and (iii),  let $d=d_\rho$  in Part (ii), and   $d=d_\gamma$  in Part (ii); note that $ d \ge  N$ in both cases. 
We set $Y_{a}=[A]^\La_{{a} {d}}\setminus  [A]^\La_{({a}-1) {d}}\subset B$ for ${a}=1,2,\ldots, 3k$;  note  $3k d\le \frac  {\rho^\La(A,B)}2$ in both cases.
 
 The set $M$ consists of  $s$ clusters  where $2\le s \le k$, so $N\ge 2$.  Each cluster has length $\le N-1$, so it can intersect at most two of the $Y_{a}$'s (as $ d \ge  N$), hence $M$ can intersect at most $2k$ of the distinct $Y_{a}$'s.  Thus there exists ${a}_* \in \set{1,2,\ldots, 3k-1}$ such that 
 \be\label{eq:M}
 M\cap \pa{Y_{{a}_*}\cup Y_{{a}_*+1}}=\emptyset,
 \ee  
 and $M_1=  M\cap [A]^\La_{{({a}_*-1)}d}\ne \emptyset$ since  $A\cap M\neq\emptyset$.

 To prove Part (ii) with $d=d_\rho$,  set $K=  [A]^\La_{{a}_* d_\rho}\subset B$. Then  $M_1= M\cap K\ne \emptyset$ since $A\cap M\neq\emptyset$,  and $M_2= M\cap(\La \setminus K)\ne \emptyset $ as $\rho^\La(A,B) \le 2\gamma_A(M)$ by hypothesis.  Moreover,   \eq{eq:distMT} holds due to \eqref{eq:M}.
 
To prove Part (iii) with $d=d_\gamma$,  let  $K$  be given in \eq{eq:K}. Then  $M_1= M\cap K\ne \emptyset$ since $A\cap M\neq\emptyset$,  and $M_2= M\cap(\La \setminus K)\ne \emptyset $ as $\rho^\La(A,B) \le 2\gamma_A(M)$ by hypothesis.  Moreover,   \eq{eq:distMT} holds due to \eqref{eq:M}.    
\end{proof}

  Motivated by Lemma \ref{lem:partM},  given $A\subset B\subset \La$ with $ 8kN< \rho^\La  (A,B)<\infty$ we decompose $\mathcal S^{\La,A}_{N,k}$ into three distinct groups:  
\begin{enumerate}
\item  \emph{Small $\gamma_A(M)$:}  $M\in \cG_1^{\La,N}(A,B)$   if   $2\gamma_A(M)\le 8kN < \rho^\La  (A,B) $.
\item \emph{Large $\gamma_A(M)$:}  $M\in \cG_2^{\La,N}(A,B)$   if $ 8kN< \rho^\La(A,B)\le 2 \gamma_A(M)$.
\item \emph{Intermediate  $\gamma_A(M)$:}  $M\in \cG_3^{\La,N}(A,B)$    if     $ 8kN< 2\gamma_A(M)<  \rho^\La(A,B)$ .
\end{enumerate}    
 Note that for $ 8kN< \rho^\La  (A,B)<\infty$ we have 
\beq\label{cGi}
\chi_N^\La  Q^\Lambda_{\le k} P_-^A= \sum_{1=1}^3 \pi_{ \cG_i^{\La,N}(A,B)}, \qtx{where} \pi_{ \cG_i^{\La,N}(A,B)},=  \sum_{M\in \cG_i^{\La,N}(A,B)}  \pi_M .
\eeq

\subsection{Decoupling revisited} \label{subsec:Ggroups}  
We will need to estimate   $\chi^\La_N Q^\Lambda_{\le k}P_-^AR_E^\Lambda P_+^{B}Q^\Lambda_{\le k}$.
If $ 8kN\ge\rho^\La  (A,B)$ we use \eq{eq:aprLD1}.  If $ 8kN< \rho^\La  (A,B)$,
  we note that
\beq
\pi_{M}\chi^\La_{N}Q^\Lambda_{\le k}P_-^A R^{\Lambda}_{E}P_+^B Q^\Lambda_{\le k}=\pi_{M}P_-^A R^{\Lambda}_{E}P_+^B Q^\Lambda_{\le k}  \qtx{for}  M \in S^{\La,A}_{N,k}.
\eeq
We will use  different strategies for $M\in \cG_i=\cG_i^{\La,N}(A,B)$, $i=1,2,3$.

If  $M\in \cG_1$,  we use the decoupling argument of Section \ref{subsec:dec}, getting \eq{eq:keydec}  with $K=[A]^\La_{8kN}$.  The estimation for the expression in first parenthesis in  \eq{eq:keydec} will be performed using directly the a-priori estimate \eqref{eq:aprLD} and \eq{smallgamma}.   (No  energy reduction.) This yields  exponential decay in $\gamma_A(M)$ for these type of contributions and the sub-harmonicity argument concludes the analysis.

To handle $M\in \cG_2$, we   consider $K, M_1,M_2$  as in   Lemma \ref{lem:partM}(ii), set $S=[\partial K]^\La_{d_\gamma-1} $, and note that
\be\label{eq:G2in}
{\pi_M}P_-^A R^{\Lambda}_{E}P_+^B Q^\Lambda_{\le k}={\pi_M}P_+^{S}P_-^{K}P_-^{K^c}P_-^A R^{\Lambda}_{E}P_+^B Q^\Lambda_{\le k}.
\ee 
Using $M_1 \subset B$, we get 
\be\label{eq:G_2}
P_+^{S} P_-^{ K}P_-^{K^c} R_E^\Lambda P_+^{B}=-\pa{P_+^{S}P_-^{K}P_-^{K^c} R^{K,K^c}_E P_-^{\partial^\La K}}\Gamma^\K P_-^{\partial^\La K} R^\Lambda_E P_+^{B}.
\ee 
  The expression in parenthesis is estimated  by reduction to lower energies $E'\in I_{\le k-1}$, allowing the use of the  induction hypothesis (in $k$) together   with the  estimate  \eq{eq:distMT} to obtain  exponential decay in $\rho^\La(A,B)$.

If $M\in \cG_3$,   we use a decoupling based on Lemma \ref{lem:partM}(iii), get   exponential decay in $\gamma_A(M)$ from the induction  hypothesis (in $k$), 
and the sub-harmonicity argument concludes the analysis.

\subsection{Reduction to lower energies}
We first observe that $P_-^A R^{\Lambda}_{E}P_+^B=P_-^A \what R^{\Lambda}_{0,E}P_+^B$ decays exponentially in $\rho^\La (A,B)$ for $E\le \tfrac34\tfd$ due to \eqref{CT} with $k=0$,
that is, Theorem \ref{thm:maintech} holds for k=0.   Suppose now that we already established \eqref{eq:mainbnd} for all energies $E\in I_{\le k-1}$ and we want to push the allowable energies to the interval $I_{\le k}$. The principal idea here is to observe that if $\emp \ne K \subsetneq \La$,  then we have the nontrivial decoupling   $H^{K,K^c}= H^{K} + H^{K^c}$, and $R_E^{K,K^c}$ can be decomposed as 
\be
R_E^{K,K^c} =\sum_{\nu\in \sigma (H^{K^c})} {R_{E-\nu}^{K}\otimes \pi_{\kappa_\nu}},
\ee
where   $\set{\kappa_\nu}_{\nu \in \sigma (H^{K^c})}$  is an orthonormal basis for  $\cH_{K^c}$  that diagonalizes $H^{K^c}$: $H^{K^c}\kappa_\nu=\nu\kappa_\nu$.  In particular, if $K_1\subset K $ and $K_2\subset K^c$, we deduce that
 \be\label{eq:enered}
P_-^{K_1}P_-^{K_2}R_E^{K,K^c}  =\sum_{\nu\in \sigma (H^{K^c})\cap[\fd,\infty)} \pa{P_-^{K_1}R_{E-\nu}^{K}}\otimes \pa{P_-^{K_2}\pi_{\kappa_\nu}},
\ee 
since $P_-^{K_2}\pi_{\kappa_0}=0$, and we have $\min_{\nu\in \sigma (H^{K^c})\setminus \set{0}}\,\nu\ge \fd$. This is exactly the type of setup we have in \eqref{eq:G2in}--\eqref{eq:G_2}. It means that the  factor $P_-^{K_1}P_-^{K_2}$ allows us effectively to lower the  energy $E\in I_{\le k}$ to  $E-\nu\in I_{\le k-1}$ and therefore use the induction hypothesis to obtain  exponential decay (we of course still need to control the summation over $\nu$ on the right hand side of \eqref{eq:enered}).

\section{Proof of the main theorem}\label{sec:prmain}

In the section we prove Theorem \ref{thm:maintech}.  We fix $\Delta_0> 5$ and $\lambda_0>0$, and assume 
$\Delta \ge \Delta_0$ and $\lambda \ge \lambda_0$.  As discussed in Remark~\ref{remarkDelta5}, the argument can be modified for  $\Delta_0>1$.

The proof proceeds by induction on $k$.  Theorem \ref{thm:maintech} holds for $k=0$, since in this case  \eqref{eq:mainbnd}  follows from \eqref{CT} with $F_0=C_0$,
$\xi_0=0$ and $\theta_0=m_0$ as  $P_-^A R_E^\La= P_-^A R_{0,E}^\La$. 
Given $k\in \N$, we  assume the theorem holds for $k-1$, and we will prove the theorem holds for $k$.

We now fix $k\in \N$ and   $\La \subset \Z$,  finite and  nonempty . We also fix $A\subset B \subset \La$, where
$A$ is a  nonempty subset  connected  in $\La$;  it follows that $[A]^\La_p$ is also connected in $\La$ and
$\abs{]A[^\La_{p}}\le 2$  for all $p\in \Z$.

To  derive the  bound \eqref{eq:mainbnd} from Lemma \ref{lem:inred}(i) we will  estimate
 $\E \pa{\norm{F^\La_{p,q}(E,A)}_{HS}^s}$ for $p,q= -\abs{A},  -\abs{A}+1,\ldots, \abs{\La}$ for $E\in I_{\le k}$, where $F^\La_{p,q}(E,A)$ is given in \eq{Fpq}.
The estimate \eq{eq:aprAB} gives the a priori  bound ($F_{p,q}=F^\La_{p,q}(E,A)$)
  \be\label{eq:aprABpq}
\E \norm{F_{p,q}}_{HS}^s\le C\lambda_0^{-s}k ^s\abs{\Lambda}^{2sk+2}.
\ee

Since $F_{p,q}=F_{q,p}^*$, we may assume $ p\le q$. If $p=q$ we use \eq{eq:aprABpq}, if $p<q$ we note that
\be\label{Fpqest}
 \norm{F_{p,q}}_{HS}&\le \norm{Q^\Lambda_{\le k}P_-^{ ]A[^{\Lambda}_p}  R^{\Lambda}_{E} P_+^{[A]^{\Lambda}_q}Q^\Lambda_{\le k}}_{HS}
\le \sum_{j\in ]A[^{\Lambda}_p}\norm{Q^\Lambda_{\le k}\cN_j  R^{\Lambda}_{E} P_+^{[j]^{\Lambda}_{q-p-1}}Q^\Lambda_{\le k}}_{HS},
\ee
where we used $[ ]A[^{\Lambda}_{p}]^{\Lambda}_{q-p-1}\subset [A]^{\Lambda}_q $ for $p<q$. 

For $r\in \N^0$ and  $E\in I_{\le k}$ we set
\be\label{eq:f(t)}
f^{\La}(k,E,r)=\max_{\Theta \subset\La }\max_{j\in \Theta}
\E\pa{\norm{Q_{\le k}^{\Theta}   \cN_j R^{\Theta}_E P_+^{[j]^{\Theta}_r}Q_{\le k}^{\Theta}}_{HS}^s},
\ee
and prove the following lemma.

\begin{lemma}\label{lem:F_{p,q}} Let $k\in \N$, $s\in (0,\frac 13)$,  and assume  Theorem \ref{thm:maintech} holds for $k-1$. 
Then there exist constants $D_k,C_k, \zeta_k , m_k >0$    (depending on $k$, $\Delta_0$, $ \lambda_0$ and  $s$), such that  such that, for all $\Delta \ge \Delta_0$ and $\lambda \ge \lambda_0$ with  $\lambda \Delta^2\ge D_k$,  $\Lambda\subset  \Z$ finite,  energy $E\in I_{\le k}$, and $r\in \N^0$,  we have
\be\label{eq:F_{p,q}}
f^{\La}(k,E,r)\le C_k\abs{\Lambda}^{\zeta_k} \e^{-m_k r}.
  \ee 
 \end{lemma}

To finish the proof of the theorem, we assume that $\Delta \ge \Delta_0$ and $\lambda \ge \lambda_0$ with  $\lambda \Delta^2\ge D_k$ as in the lemma. Then, since $\E \pa{\norm{F_{p,q}}}_{HS}^s\le 2 f^\La(k,E, \abs{q-p}-1)$ for $\abs{q-p}\ge 1$, 
 and we have \eq{eq:aprABpq} for 
$q=p$, we obtain
\be\label{EFsum}
&\E\pa{ \sum_{p=-\abs{A}}^{\abs{\Lambda}} \sum_{q=-\abs{A}}^{\abs{\Lambda}}\e^{-m_0(p)_+}\e^{-m_0\pa{t-q-1}_+}\norm{F_{p,q}}}^s   \le   C_k \abs{\La}^{\zeta_k} \e^{-s m_k t}.
\ee

The estimate \eqref{eq:mainbnd}   now follows from \eq{Fpq} and \eq{EFsum} (recall \eq{defrhoLa}), so Theorem \ref{thm:maintech} holds for $k$.

To complete the proof of  Theorem \ref{thm:maintech}  we need to prove Lemma \ref{lem:F_{p,q}}.  To do so we need the following lemma.

\begin{lemma}\label{lem:maintech}
 Let $k\in \N$, $s\in (0,\frac 13)$,  and assume  Theorem \ref{thm:maintech} holds for $k-1$. 
Then there exist constants $C_k, \zeta_k , m_k >0$    (depending on $k$, $\Delta_0$, $ \lambda_0$ and  $s$),   such that, for all $\Delta \ge \Delta_0$ and $\lambda \ge \lambda_0$,  $j \in \Lambda\subset  \Z$ finite,   energy $E\in I_{\le k}$,  $N\in \N$, and $r\in \N^0$ such that  $8kN <  r$ , we have
\be\label{GNr123}
G^\La_N(r)&=  \E\pa{\norm{\chi_N^{\La}Q_{\le k}^{\Lambda} \cN_j R^{\Lambda}_E P_+^{[j]^\La_r}Q_{\le k}^{\Lambda}}_{HS}^s}\\& \le C_k \pa{ \abs{\La}^{\zeta_k} \e^{-m_k  r}  + \e^{-   m_k  N}  \pa{\lambda \Delta^2}^{-s}  \sum_{p=0}^{r}  \e^{-m_k \pa{r-p}} f_N^\La(p)}.
\ee

\end{lemma}

\begin{proof}

Let $k\in \N$, $s\in (0,\frac 13)$,  and assume  Theorem \ref{thm:maintech} holds for $k-1$.  Let $j \in \Lambda\subset  \Z$ finite and  $E\in I_{\le k}$,  $N\in \N$, and $r\in \N^0$ such that  $8kN <  r$.
 Let $G^\La_N(r)$  be as in \eq{GNr123}. It follows from \eq{cGi}, setting  $\cG^N_i =\cG_i^{\La,N}(\set{j}, [j]^\La_r)$, $i=1,2,3$   (see Section~\ref{clusterclass}),  that
\be\label{sumGi}
G^\La_N(r)&\le \sum_{i=1}^3 G_i(r), \sqtx{where}
G_i(r)=G^{\La,N}_i(r)=  \E\pa{\norm{\pi_{ \cG^N_i }\cN_j R^{\Lambda}_E P_+^{[j]^\La_r}Q_{\le k}^{\Lambda}}_{HS}^s}.
\ee

To estimate $G_1(r)$, we use
\eqref{eq:keydec} with $M= [j]_{4kN}^\La$ and  $K= [j]_{6kN}^\La$,  \eq{eq:weak1-1}, and \eq{PGamma1},  obtaining
\be\label{eq:2expect1g4}
&G_1(r) \le 
C \pa{\lambda \Delta^2}^{-s} \,\E_K\pa{\norm{Y}_{HS}^s}\,\E_{\pa{[K]_1^\La}^c} \pa{\norm{Z}_{HS}^s};\\ & \quad Y:= \chi_N^{K} Q^K_{\le k}{P_+^{K\setminus M}R_E^{K}P_-^{\partial_{in}^{\La} K}},\quad Z:=P_-^{\partial_{ex}^{\La}  [K]_1} R_E^{\pa{[K]_1^\La}^c}P_+^{[j]^\La_r\cap\pa{[K]_1^\La}^c}Q^{\pa{[K]_1^\La}^c}_{\le k}\chi_N^{\pa{[K]_1^\La}^c}.
\ee

 To estimate $\E_K\pa{\norm{Y}_{HS}^s}$, note that 
\be \label{estYHS1}
 \norm{Y}_{HS} \le \sum_{u\in \partial_{in}^{\La}  K} {\norm{Y_u}_{HS}},   \sqtx{where} Y_u= \chi_N^{K} Q^K_{\le k}P_+^{K\setminus M}R_E^{K} \cN_u, \sqtx{and}  \abs{\partial_{in}^{\La}  K}\le 2. 
\ee
Using \eqref{eq:inredaHS} and $\rho^K(\partial_{in}^{\La}  K,K\setminus M) \ge 2kN$,   for $u\in \partial^{\La} _{in} K$  we get
\be
\E_K\pa{\norm{Y_u}_{HS}^s}&\le C_k^s\pa{\abs{K}^{sk}\e^{-sm_0 2kN}+  \sum_{q=-1}^{\abs{K}}  \e^{-sm_0\pa{q}_+} \E\pa{\Big\|{\chi_N^{K}  Q^K_{\le k}P_-^{]u[^K_{q}}R^K_{E}P_+^{K\setminus M} Q^K_{\le k}}^s_{HS}\Big\|}}\\
&  \le  C_{k,s}\pa{ \abs{K}^{2sk+1} \e^{-sm_02kN}+
2 \sum_{q=-1}^{2kN-1}  \e^{-sm_0\pa{q}_+} f_N^K(2kN-q-1)}
\\
& \le  C_{k,s} \e^{-   m_{0,k}^\pr k  N} . 
\ee
where we used the  a priori bounds  \eq{eq:aprAB} and  \eq{eq:aprLD}.

Similarly,
\be\label{estZ1}
\norm{Z}_{HS}\le\sum_{u\in \partial_{ex}^{\La}  [K]_1} {\norm{Z_u}_{HS}},   \sqtx{where} Z_u= \cN_u R_E^{\pa{[K]_1^\La}^c}P_+^{[j]^\La_r\cap \pa{[K]_1^\La}^c}Q^{\pa{[K]_1^\La}^c}_{\le k}\chi_N^{\pa{[K]_1^\La}^c},  
\ee
and  $ \abs{\partial_{ex}^{\La}  [K]_1}\le 2$.   Using \eqref{eq:inredaHS}, for $u\in \partial_{ex} [K]_1$  we get
\be\label{ZcG1}
&\E_{\pa{[K]_1^\La}^c}\pa{\norm{Z_u}_{HS}^s}
 \le C_k^s\left(\abs{\La}^{sk}\e^{-sm_0 (r-6kN-2)}+ \right.  \\
 & \quad \left. \sum_{q=-1}^{\abs{\La}}  \e^{-sm_0\pa{q}_+} \E\Big({\Big\|{\chi_N^{{\pa{[K]_1^\La}^c}}  Q^{\pa{[K]_1^\La}^c}_{\le k} P_-^{]u[^{\pa{[K]_1^\La}^c}_{q}} R_E^{\pa{[K]_1^\La}^c} P_+^{[j]^\La_r\cap \pa{[K]_1^\La}^c} Q^{\pa{[K]_1^\La}^c}_{\le k}}\Big\|^s_{HS}}\Big)\right)\\
&  \le  C_k^s\pa{ \abs{\La}^{2sk+2} \e^{-sm_0 (r-6kN-2)}+
 \sum_{q=-1}^{r-6kN-3}  \e^{-sm_0\pa{q}_+} f_N^\La(r-6kN-q-3)}\\
 &  = C_k^s \pa{ \abs{\La}^{2sk+2} \e^{-sm_0 (r-6kN-2)}+
 \sum_{p=0}^{r-6kN-2}  \e^{-sm_0\pa{r-p-6kN -3}_+} f_N^\La(p)}.
\ee

Combining \eq{eq:2expect1g4}-\eq{ZcG1}, we get 
\be\label{gcG1}
&G_1(r)   \le  C \pa{\lambda \Delta^2}^{-s}  \e^{-   m_k^\pr k  N}\Big({ \abs{\La}^{2sk+2} \e^{-m_k^\pr  r}+
 \sum_{p=0}^{r}  \e^{-m_k^\pr \pa{r-p}} f_N^\La(p)}\Big),
\ee
for an appropriate $m_k^\pr >0$.

To estimate $G_2(r)$,  we note that it follows from  Lemma \ref{lem:partM}(ii), letting 
\be\label{eq:Kse} 
K(a)=[ j]^{\La}_{ad_\rho}\qtx{and} S(a)=[\partial^{\La} K(a)]^{\La}_{d_\rho-1}  \qtx{for} a\in\N,
\ee
that
\be\label{G2rrho}
G_2(r) \le \sum_{a=1}^{3k-1} G\up{a}_2(r), \quad G\up{a}_2(r)=   \E\pa{\norm{ \chi_N^{\La} Q^\La_{\le k}P_+^{S(a)}P_-^{K(a)}P_-^{(K(a))^c}\cN_j R^{\Lambda}_E P_+^{[j]^\La_r}Q_{\le k}^{\Lambda}}_{HS}^s}.
\ee

To estimate $ G\up{a}_2(r)$, we use \eq{eq:G2in} and \eq{eq:G_2}, the Cauchy-Schwarz inequality, and H\"older's inequality (recall $3s <1$),  to get (we mostly omit $a$ from the notation)
\be\label{eq:YZ}
G\up{a}_2(r) &\le C \Delta^{-s}  \pa{\E
\norm{Y}^{2s}}^{1/2}\pa{\E\norm{Z}^{2s}_{HS}}^{1/2}\\  & \le  C \Delta^{-s} \pa{\E\norm{Y}^{s}}^{1/4} \pa{\E\norm{Y}^{3s}}^{1/4}\pa{\E\norm{Z}^{2s}_{HS}}^{1/2} ,
\ee
where
\be\label{YZG2}
Y=\chi_N^{\La} Q^\La_{\le k}P_+^{S(a)}P_-^{K(a)}P_-^{{(K(a))^c}}\cN_jR_E^{K(a), (K(a))^c}P_-^{\partial ^{\La} K(a)} \sqtx{and} Z=P_-^{\partial^{\La} K(a)} R^\Lambda_E  P_+^{[j]^\La_r}Q_{\le k}^{\Lambda}\chi_N^\La.
\ee 

It follows immediately from  \eqref{eq:inredaHSE}   that
\be\label{normZ2s}
\E\norm{Z}^{2s}_{HS}\le  C \abs{\Lambda}^{4sk+3} \qtx{and} \E\norm{Y}^{3s}\le C   \abs{\Lambda}^{6sk+3},
\ee
where we used  $\abs{\partial^{\La} K(a)}\le 4$ since $K(a)$ is  connected, and hence  we have
\be\label{eq:YYZ}
G\up{a}_2(r) \le    C \Delta^{-s}  \abs{\La}^{\frac 72 sk+\frac 94}    \pa{\E\norm{Y}^{s}}^{1/4} .
\ee

To estimate $E\norm{Y}^{s}$, we use ( the dependence on $a$ is being ommitted)
\be\label{eq:splitx1'}
\norm{Y}\le \sum_{x\in\partial^{\La} K}\norm{Y_x}, \qtx{with}  Y_x=\chi_N^{\La} Q^\La_{\le k}P_+^{S}P_-^{K}P_-^{K^c}\cN_jR_E^{K,K^c}\cN_x. 
\ee

We  consider first  the case $x\in\partial^{\La}_{in}K$ . Using \eqref{eq:enered}, we can further decompose $Y_x$ as  
  \be
Y_x 
 =\sum_{ \nu\in \sigma (H^{K^c})\cap [\fd,\infty)} Y_{x,\nu}, \quad  Y_{x,\nu}= \chi_N^{\La} Q^\La_{\le k}P_+^{S}P_-^{K}P_-^{K^c}\cN_j\pa{R_{E-\nu}^{K}\otimes \pi_{\kappa_\nu}}\cN_x .
\ee
Note that 
\be\label{eq:splY1}
\norm{Y_x}=\max_\nu\norm{Y_{x,\nu}}\le \sum_{ \nu\in \sigma (H^{K^c})\cap [\fd,k(\fd))} \norm{Y_{x,\nu}}+\max_{ \nu\in \sigma (H^{K^c})\cap [k(\fd),\infty)} \norm{Y_{x,\nu}}.
\ee
Clearly, we can bound 
\be
\norm{Y_{x,\nu}}\le \norm{P_+^{S} \pa{R_{E-\nu}^{K}\otimes \pi_{\kappa_\nu}}\cN_x}\le \norm{P_+^{S\cap K}R_{E-\nu}^{K} \cN_x}.
\ee

For $\nu\ge\fd$, we have $E-\nu \in  I_{\le k-1} $ for $E\in I_{\le k}$ (recall \eqref{Ikle}).  For  $\nu\in \sigma (H^{K^c})\cap [\fd,k(\fd)) $, \
we  use the induction hypothesis for Theorem \ref{thm:maintech}
and the statistical independence of $H^{K^c}$ and $\set{\omega_i}_{i\in K}$ to conclude that
\be\label{eq:frscr1}
\E{\norm{Y_{x,\nu}}^{s}}\le  \E_{K}{{\norm{P_+^{S\cap K}R_{E-\nu}^{K} \cN_x}}^{s}} \le C_{k-1}\abs{\Lambda}^{\xi_{k-1}}\e^{-\theta_{k-1} \frac r {6k}},
\ee  
 where  we used \eqref{eq:distMT}.

For $\nu\in \sigma (H^{K^c})\cap [k(\fd),\infty)$, $E-\nu \le  \frac 34 \tfd$, and in this case 
\be
P_+^{S\cap K}R_{E-\nu}^{K}  \cN_x= P_+^{S\cap K}\what R_{E-\nu}^{K} \cN_x,
\ee so it follows from \eqref{CT} with $k=0$, using 
 \eqref{eq:distMT},  that 
\be\label{eq:frscr2}
  \norm{Y_{x,\nu}} \le  C_0 \e^{-m_{0}  \frac r {6k}}. 
\ee
Using \eqref{eq:splY1}, \eqref{eq:frscr1},  \eqref{eq:frscr2}, and \eq{trkH},  we get
 \be\label{eq:splitxf'}
 \E\norm{Y_{x}}^{s}& \le C\abs{\Lambda}^{\xi_{k-1}}\e^{-\theta_{k-1} \frac r {6k}}\tr  \chi_{ [\fd,k(\fd))}(H^{K^c})           \le C_k\abs{\Lambda}^{\xi_{k-1}+2k}\,\e^{-\frac {\theta_{k-1}}{6k} r }.
  \ee
Similar considerations show that the estimate \eq{eq:splitxf'} holds also for $x\in\partial_{ex}^{\La}K $.

Combining \eq{eq:splitx1'} and \eq{eq:splitxf'} and  recalling $\abs{\partial^{\La} K}\le 4$,  we get
\be\label{eq:Y2s1}
\E\norm{Y}^{s}\le C_k  \abs{\Lambda}^{\xi_{k-1}+2k}\,\e^{-\frac {\theta_{k-1}}{6k} r }.
\ee

Combining \eqref{eq:YYZ} and \eq{eq:Y2s1},  we see that
\be\label{Ga2r}
G\up{a}_2(r) \le  C_k \Delta^{-s}\abs{\Lambda}^{\zeta_k}\,\e^{-\frac {\theta_{k-1}}{24 k} r }.
\ee

It now follows from  \eq{G2rrho} and   \eq{Ga2r} that
\be\label{gcG2}
G_2(r) \le  C_k  \Delta^{-s}\abs{\Lambda}^{\zeta_k}\,\e^{-\frac {\theta_{k-1}}{24 k} r }\le C_k \Delta^{-s}\abs{\Lambda}^{\zeta_k}
\e^{- \theta_{k-1}^{\prr} r }.
\ee

To estimate $ G_3(r)$, given  $4kN< \gamma< \frac r 2$, we let $ d_\gamma:= \fl{ \frac{ \gamma}{3k}}$.  Given $a\in  \set{1,2,\ldots,3k-1}
$, 
we let $K(a,\gamma)$ be as in \eq{eq:K} with $A=\set{j}$, and let  $K_1(a,\gamma)=[j ]^\La_{a d_\gamma}$,    the connected component of $K(a,\gamma)$ that contains $j$. We also set $K_2(a,\gamma)=K(a,\gamma)\setminus K_1(a,\gamma)$, $S(a,\gamma)=[\partial^\La K(a,\gamma)]^\La _{d_\gamma-1}$, and $T(a,\gamma)=[j]^\La_{\gamma_{\set{j}}(M)}$. It follows from  Lemma \ref{lem:partM}(iii) that
\be
\pi_M=\pi_M P_+^{T(a,\gamma_{\set{j}}(M))}P_+^{S(a,\gamma_{\set{j}}(M))}P_-^{K_1(a,\gamma_{\set{j}}(M))}P_-^{K_2(a,\gamma_{\set{j}}(M))},
\ee
{for some} $a\in \set{1,2,\ldots,3k-1}$, 
and hence
\be\label{G3rrho}
G_3(r) &\le \sum_{\gamma=4kN+1}^{\fl{\frac r 2}} \sum_{a=1}^{3k-1} G\up{a,\gamma}_3(r), \qtx{where}    \\
G\up{a,\gamma}_3(r)& =   \E\pa{\norm{ \chi_N^{\La} Q^\La_{\le k}P_+^{T(a,\gamma)}P_+^{S(a,\gamma)}P_-^{K_1(a,\gamma)}P_-^{K_2(a,\gamma)}\cN_j R^{\Lambda}_E P_+^{[j]^\La_r}Q_{\le k}^{\Lambda}}_{HS}^s}.
\ee

To estimate $G\up{a,\gamma}_3(r)$, we start with the
 following analogue of \eqref{eq:2expect1g4}  (we mostly omit $(a,\gamma)$ from the notation):
\be\label{eq:2expect1g4?}
&G\up{a,\gamma}_3(r)\le 
C \pa{\lambda \Delta^2}^{-s} \,\E_K\pa{\norm{Y}^s}\,\E_{[K]_1^c} \pa{\norm{Z}_{HS}^s};\\ & \quad Y:= \chi_N^{K} Q^K_{\le k}P_+^{T\cap K}P_+^{S}P_-^{K_1}P_-^{K_2}R_E^{K}P_-^{\partial^\La_{in} K},\\ & \quad Z:=P_-^{\partial_{ex}^\La [K]_1} R_E^{\pa{[K]_1^\La}^c}P_+^{[j]^\La_r\cap{\pa{[K]_1^\La}^c}}Q^{{\pa{[K]_1^\La}^c}}_{\le k}\chi_N^{{\pa{[K]_1^\La}^c}}.
\ee

 Proceeding exactly as in \eq{estZ1}-\eq{ZcG1}, we get
\be\label{ZcG133}
\E_{\pa{[K]_1^\La}^c}\pa{\norm{Z}_{HS}^s} \le  C_k^s\Big({ \abs{\La}^{\xi_k}  \e^{-sm_0 (r-\gamma -d_\gamma)}+
 \sum_{p=0}^{r-(\gamma +d_\gamma)-2}  \e^{-sm_0\pa{r-p-(\gamma +d_\gamma) -3}_+} f_N^\La(p)}\Big).
\ee
 We estimate  $E\norm{Y}^{s}$ similarly to \eq{eq:splitx1'}-\eq{eq:Y2s1}. We have
\be\label{eq:splitx1'33}
\norm{Y}\le \sum_{x\in \partial_{in} K}\norm{Y_x}, \qtx{where}  Y_x=\chi_N^{K} Q^K_{\le k}P_+^{T\cap K}P_+^{S}P_-^{K_1}P_-^{K_2}R_E^{K}\cN_x .
\ee 
We  consider first  the case $x=x_i\in\partial_{in}{\pa{[K]_1^\La}^c}K_i$,  $i\in \set{1,2}$ , and  $i'=\set{1,2}\setminus \set{i}$.  Using \eqref{eq:enered}, we can further decompose $Y_x$ as  
  \be
Y_{x_i} 
 =\sum_{ \nu\in \sigma (H^{K_{i'}})\cap [\fd,\infty)} Y_{x_i,\nu}, \quad  Y_{x_i,\nu}= P_+^{S} P_-^{K_{i'}}  \pa{R_{E-\nu}^{K_i}\otimes \pi_{\kappa_\nu}}\cN_x .
\ee
Note that 
\be\label{eq:splY133}
\norm{Y_{x_i} }=\max_\nu\norm{Y_{{x_i} ,\nu}}\le \sum_{ \nu\in \sigma (H^{K_{i'}})\cap [\fd,k(\fd))} \norm{Y_{{x_i} ,\nu}}+\max_{ \nu\in \sigma (H^{K_{i'}})\cap [k(\fd),\infty)} \norm{Y_{{x_i} ,\nu}}.
\ee
Clearly, we can bound 
\be
\norm{Y_{{x_i} ,\nu}}\le \norm{P_+^{S} \pa{R_{E-\nu}^{K_i}\otimes \pi_{\kappa_\nu}}\cN_{x_i} }\le \norm{P_+^{S\cap K_i}R_{E-\nu}^{K_i} \cN_{x_i} }.
\ee

For $\nu\ge\fd$, we have $E-\nu \in  I_{\le k-1} $ for $E\in I_{\le k}$ (recall \eqref{Ikle}).  For  $\nu\in \sigma (H^{K_{i'}})\cap [\fd,k(\fd)) $, \
we  use the induction hypothesis for Theorem \ref{thm:maintech}
and the statistical independence of $H^{K_{i'}}$ and $\set{\omega_i}_{i\in K_i}$ to conclude that
\be\label{eq:frscr133}
\E{\norm{Y_{x,\nu}}^{s}}\le  \E_{K_i}\norm{P_+^{S\cap K_i}R_{E-\nu}^{K_i} \cN_{x_i}}^s \le C_{k-1}\abs{K_i}^{\xi_{k-1}}\e^{-\theta_{k-1} d_\gamma}\le C_{k-1}\abs{\gamma}^{\xi_{k-1}}\e^{-\theta_{k-1}d_\gamma}.
\ee

For $\nu\in \sigma (H^{K_{i'}})\cap [k(\fd),\infty)$, $E-\nu \le  \frac 34 \tfd$, and in this case \[P_+^{S\cap K_i}R_{E-\nu}^{K_i}  \cN_x= P_+^{S\cap K_i}\what R_{E-\nu}^{K_i}  \cN_{x_i},\] so it follows from \eqref{CT} with $k=0$  that 
\be\label{eq:frscr233}
  \norm{Y_{x_i,\nu}} \le  C_0 \e^{-m_{0}  d_\gamma}. 
\ee
Using \eqref{eq:splY133}, \eqref{eq:frscr133},  \eqref{eq:frscr233}, and \eq{trkH},  we get
 \be\label{eq:splitxf'33}
 \E\norm{Y_{x_i}}^{s}& \le C\gamma^{\xi_{k-1}}\e^{-\frac {\theta_{k-1}}{3k} \gamma}\tr  \chi_{ [\fd,k(\fd))}(H^{K_{i'}})           \le C_k\gamma^{\xi_{k-1}+2k}\,\e^{-\frac {\theta_{k-1}}{3k} \gamma}.
 \ee

Combining \eq{eq:splitx1'33} and \eq{eq:splitxf'33} and  recalling $\abs{\partial_{in}^\La K_i}\le 4$,  we get
\be\label{eq:Y2s133}
\E\norm{Y}^{s}\le C_k  \abs{\gamma}^{\xi_{k-1}+2k}\,\e^{-\frac {\theta_{k-1}}{3k} \gamma}  \le  C_k\e^{- \theta^{\pr}_{k-1} \gamma}.
\ee

Combining \eq{eq:2expect1g4?},  \eq{ZcG133}, and  \eq{eq:Y2s133}, we get
\be\label{G3rrhoga}
G\up{a,\gamma}_3(r)\le C_k \pa{\lambda \Delta^2}^{-s}\e^{- \theta^\pr_{k-1} \gamma}
 \Big({ \abs{\La}^{2sk+2}  \e^{-s \theta^\pr_{k-1} r}+
 \sum_{p=0}^{r }  \e^{-s \theta^\pr_{k-1} \pa{r-p}}f_N^\La(p)  }\Big).
\ee
It follows from \eq{G3rrho} and \eq{G3rrhoga} that
\be\label{gcG3}
G_3(r)
 \le C_k \pa{\lambda \Delta^2}^{-s}  \e^{-\what  m N}\Big({ \abs{\La}^{2sk+2}  \e^{-s\theta^\pr_{k-1}r}+
 \sum_{p=0}^{r}  \e^{-s \theta^\pr_{k-1} \pa{r-p}}f_N^\La(p)  }\Big).
\ee

Putting together \eq{sumGi}, \eq{gcG1}, \eq{gcG2} , and \eq{gcG3}, we obtain \eq{GNr123}.
\end{proof}

We can now prove Lemma \ref{lem:F_{p,q}}.

\begin{proof}[Proof of Lemma \ref{lem:F_{p,q}}]
For $\La \subset \Z$ finite,  $E\in I_{\le k}$,  $N\in\N$,  and $r\in \N^0$ , we set
  \be \label{eq:subhfanew}
f^\La_{N}(r)= f_N^{\La}(k,E,r)=\max_{\Theta \subset\La }\max_{j\in \Theta}
\E\pa{\norm{\chi_N^\Theta  Q_{\le k}^{\Theta}   \cN_j R^{\Theta}_E P_+^{[j]^{\Theta}_r}Q_{\le k}^{\Theta}}_{HS}^s}.
\ee 
Note  that  $f^\La_N(r)$ is monotone increasing in $\La$, and it follows   from \eq{eq:aprAB} that 
\be\label{fbddNknew}
\max_{r\in \N^0} f_N^\La(r) \le  C\lambda^{-s}k ^s\abs{\Lambda}^{2sk+1}.
\ee
Moreover, if $8kN \ge r$, it follows from \eq{eq:aprLD1} that
\be\label{fNlarded}
f^\La_{N}(r)\le C_{k,s}\abs{\Lambda}^{2(sk+1)}  e^{-m_{0,\mu}r} .
\ee

If $8kN <  r$ we  use   Lemma \ref{lem:maintech}.
Since this lemma holds for arbitrary  finite subsets of $\Z$, it follows from  \eq{GNr123} that for $8kN < r$ we have
\beq
f_N^\La(r) \le  C_k \pa{ \abs{\La}^{\zeta_k} \e^{-m_k  r}  + \e^{-   m_k  N}  \pa{\lambda \Delta^2}^{-s}  \sum_{p=0}^{r}  \e^{-m_k \pa{r-p}} f_N^\La(p)},
\eeq
for all  $\La \subset \Z$ finite. Combining with \eq{fNlarded}, we get  (with possibly slightly different constants $C$, $m_k >0$,   $\zeta_k >0$)
\be\label{eq:fineq}
f^\La(r) \le  \sum_{N=1} ^{\abs{\Lambda}} f_N^\La(r) \le C\pa{ \abs{\La}^{\zeta_k} \e^{-m_k  r}  +  \pa{\lambda \Delta^2}^{-s}  \sum_{p=0}^{r}  \e^{-m_k \pa{r-p}} }.
\ee

The proof can now be completed by a standard subharmonicity argument. Let $h^\La(r)=  f^\La(r) - 2C\abs{\La}^{{\zeta_k}} \e^{-m_k \frac r2}$, and take $\Delta \ge \Delta_0$ and $\lambda \ge \lambda_0$ such that
\be\label{eq:subhara'}
2C   \pa{\lambda \Delta^2}^{-s}  \sum_{q=-\infty}^\infty \e^{-m_k \frac {\abs{q}} 2}\le1.
\ee

 Then \eqref{eq:fineq} implies that 
\be\label{eq:gtest}
h^\La(r) &\le C\abs{\La}^{{\zeta_k}} \e^{-m_k  r}-  2C\abs{\La}^{{\zeta_k}} \e^{-m_k \frac r2}\\
& \qquad 
+  C \pa{\lambda \Delta^2}^{-s}  \sum_{p=0}^{r}  \e^{-m_k \pa{r-p}} \pa{h^\La(p)+ 2C\abs{\La}^{{\zeta_k}} \e^{-m_k \frac p2}}\\
& \le C\abs{\La}^{{\zeta_k}} \pa{\e^{-m_k  r}-\e^{-m_k \frac r2}}+C \pa{\lambda \Delta^2}^{-s}  \sum_{p=0}^{r}  \e^{-m_k \pa{r-p}} h^\La(p),
\ee
for all $r\in\N^0$. In addition, it follows from \eqref{fbddNknew} that 
\be R=\sup_{r\in\N^0}h^\La (r)\le \sup_{r\in\N^0}f^\La (r)\le  C\abs{\Lambda}^{2sk+3}<\infty.
\ee

We claim that 
 $R\le 0$, which implies that \eqref{eq:F_{p,q}} holds (with different constants), finishing the proof of Lemma \ref{lem:F_{p,q}}. Indeed, suppose  that $R>0$. Then it follows from  \eq{eq:gtest} and \eq{eq:subhara'} that  
\be
R\le    C  \pa{\lambda \Delta^2}^{-s} \sup_{r\in\N^0}\pa{\sum_{p=0}^{\abs{\Lambda}}\e^{-m_k\abs{r-p}} }  R \le  C  \pa{\lambda \Delta^2}^{-s}  \pa{   \sum_{q=-\infty}^\infty \e^{-m_k\frac {\abs{q}} 2}}R \le
\tfrac 12 R,
\ee 
 a contradiction.
\end{proof}

 The proof of Theorem \ref{thm:maintech} is complete.

 \section{Quasi-locality in expectation}\label{sec:cor}
 
In this section we prove Corollary~\ref{thm:eigencor}. To do so we first extract from Theorem~\ref{thm:maintech} a probabilistic statement (cf.  \cite[Proposition 5.1]{ETV}  and \cite[Lemma 7.2]{EKS1}). 
 
 We fix $k\in \N$ and let $s, \theta_k, \xi_k$ be as in \eq{eq:mainbnd}, slightly modified so \eq{eq:mainbnd} holds with  $\rho^\La(A,B)$ substituted for $\dist_\La(A,B^c)$ (recall \eq{defrhoLa}).  
 
 We fix a finite subset  $\La$ of $\Z$ .
 Given  $\emp\ne K\subset \Lambda$,  we let $H^{K^\pr}$  be the restriction of $H^K$ to $\Ran P_-^K=\Ran \chi_{\N}(\cN_K)$, $K^c=\La \setminus K$ (we allow $K^c=\emp$), and consider  $H^{{K^\pr,K^c}} =H^{K^\pr}+ H^{K^c} $,   $ \Gamma^{{K^\pr,K^c}} =H^\La-H^{{K^\pr,K^c}} $, $R_E^{{K^\pr,K^c}} =(H^{{K^\pr,K^c}} -E)^{-1}$, operators on 
 $ \Ran P_-^K \oplus \cH_{K^c}$.   
 Given an interval $I$ and an operator $H$, we set $\sigma_I(H)=\sigma(H)\cap I$.

 We start by proving Wegner-like estimates for the XXZ model. 
 
 \begin{lemma} Let $\emp\ne K\subset \Lambda$. 
  \begin{enumerate}
  
\item   Consider the open  interval $I\subset I_k$.  Then

\be\label{Wegnerineq}
\P_K \set{\sigma_I(H^{{K^\pr,K^c}} ) \ne \emp} \le C_k \lambda^{-1} \abs{I} \abs{\Lambda}^{2k+1}.
\ee

\item   Let $0<  \delta < \frac  1 4{\nfd}$.  Then (recall  \eq{Ikle})
\be
\P\set{\dist \set{\sigma_{\what I_k}(H^{{K^\pr,K^c} }), \sigma_{\what I_k} (H^{K^c})}<\delta }\le C_k \lambda^{-1} \delta \abs{\Lambda}^{4k+1}.
\ee
\end{enumerate}

 \end{lemma}

\begin{proof} To prove Part (i), 
recall \eq{trkH}  (it applies to $H^{\pa{K^\pr,K^c}}  $), let $E_1 \le E_2 \le \ldots $ be the at most $Ck\abs{\Lambda}^{2k}$
 eigenvalues of $H^{{K^\pr,K^c} }$ in $\what  I_{\le k}$, counted with multiplicity, which we consider as functions of  $ \omega_{K}$ for fixed $\omega_{K^c}$.  Since $\cN_K \ge 1$, each  $E_n( \omega_{K})$ is a monotone function on $\R^{\abs{K}}$. Let $e=(1,1,\ldots,1)\in \R^{\abs{K}}$.    We have $E_n(\omega_{K}+te) - E_n(\omega_{K}) \ge \lambda t$ for all $t>0$ and all $n$ by the min-max principle, so we  can apply  Stollmann's Lemma \cite{St}  to get
\beq\label{WegoneeigN}
\P_{K} \{E_n({\omega_K}) \in I \} \le C \abs{I} \lambda^{-1} \abs{K}.
\eeq 

In view of \eq{trkH}, \eq{Wegnerineq} follows using \eq{WegoneeigN} for each one of the eigenvalues $E_n$. 

Part (ii) follows from Part (i) and \eq{trkH} for $H^{K^c}$, since the random variables $\omega_{K}$ and $\omega_{K^c}$ are independent.
\end{proof}

 Let $E\in \R$,  $m>0$,  $ r\in \N$, $\emp \ne K\subset \La$, and let   $H^\sharp$ denote either $H^K$ or $H^{\pa{K^\pr,K^c}}$.  Then  the operator
$H^{K^\sharp}$ is said to be $(m,E,r)$-regular if
\be \label{defregular}
&F^{K^\sharp}_E \le  \e^{-mr}  \qtx{and} \dist(E,\sigma (H^{K^\sharp})) >\e^{-mr} ,\\
& \text{where}\quad  F^{K^\sharp}_E= \max_{i\in K}F^{K^\sharp}_E(i) \qtx{with}   F^{K^\sharp}_E (i)=\norm{\cN_i R^{K^\sharp}_E P_+^{[i]_r^K}}.
\ee
In addition, consider the probabilistic event
\be\label{MSAdef}
\cF_{k}^\La (K,m,r)= \set{E\in I_k \implies \sqtx{either} \  H^{\pa{K^\pr,K^c}} \sqtx{or } H^{{K^c}} \sqtx{is}  (m, E,r)\text{-regular}}.
\ee

 \begin{lemma} \label{lemMSAprob} Let $\emp \ne K\subsetneq \La$ , and let  $r\in \N$, $r\ge\frac {18} {\theta_k}$ .
Then
 \be
\P\set{\pa{\cF_{k}^\La (K,\tfrac {\theta_k} 9,r)}^c} \le  C\abs{\La}^{\xi^\pr_k} \e^{- \frac {\theta_k} 9 r}.
 \ee
  \end{lemma}

\begin{proof}
 Let $\emp \ne K\subsetneq \La$, $r\ge   \tfrac {18} {{\theta_k}}$,  and set $m= \frac {{\theta_k}} 9$, so $\e^{mr}\ge 4$.  Let $S$ denote either the pair  ${{K^\pr,K^c}}$ or ${K^c}$, and let $S^\pr=K$ if $S={{K^\pr,K^c}} $,  or $S^\pr=K^c$ if $S= K^c$.  Consider the (random) energy sets
 \be
 D_S=\set{E\in I_k: \  F^{S}_E> \e^{-mr}}\qtx{and}
 J_S=\set{E\in I_k:  \ F^{S}_E> \e^{-2mr}},
 \ee
and the event
\be
\cJ_S = \set{\abs{J_S} > e^{-5mr}}.
\ee

Using \eq{eq:mainbnd} we get
\be
\P\set{\cJ_S}  & \le    \e^{5mr} \, \E \set{ \abs{J_S}} \le \e^{5mr} \, \E \set{ \int_{I_k} \e^{2s mr}   \pa{F^{S}_E}^s \ dE    }\\
& \le \e^{7 mr} \int_{I_k}   \sum_{i \in S^\pr} \E\set{\pa{F^{S}_E(i) }^s }\ dE    \le C_k \abs{\La}^{\xi_k+1} e^{-2mr}.
\ee  
We now consider the (random) energy set
\be
Y_S= \set{E\in I_k: \  \dist(E,\sigma (H^{S})) \le \e^{-mr} },
\ee
and claim that $D_S \subset Y_S$  on the complementary event $\cJ_S^c=  \set{\abs{J_S} \le  e^{-5mr}} $.

To see this, suppose  $\abs{J_S} \le  e^{-5mr}$ and $E\in D_S\setminus Y_S$.  Since $E\in D_S$, there exists $i\in S^\pr$ such that $F_E^S(i) >  e^{-mr}$.  Let $E^\pr \in I_k$ such that $\abs{E^\pr-E}\le 2 e^{-5mr}$.
Using $E \in Y_S$ we get $\dist(E^\pr, \sigma(H^S) >  e^{-mr}- 2 e^{-5mr}\ge \frac 1 2 e^{-mr}$. Thus, using the resolvent identity and $r\ge\frac {18} {\theta_k}$, we have
\be
F_{E^\pr}^S(i)\ge F_{E}^S(i)-\abs{E^\pr-E}\norm{R_E^S}\norm{R_{E^\pr}^S}> 
 e^{-mr}- (2e^{-5mr})  e^{mr}(2 e^{mr})\ge  e^{-2mr}.
\ee
It follows that    $[E- 2 e^{-5mr}, E+ 2 e^{-5mr}]\cap I_k \subset J_S$.
Since $\abs{I_k}\ge   2 e^{-5mr}$ as $r\ge\frac {18} {\theta_k}$ ,      we conclude that $\abs{J_S} \ge 2 e^{-5mr} > e^{-5mr}$,
a contradiction.  

We proved that  $\abs{J_S} \le  e^{-5mr} $ implies $ D_S \subset Y_S$, so $\what Y_S= I_k\setminus Y_S\subset  I_k\setminus D_S$.  In particular, outside the event $\cJ_S$, $E\in \what Y_S$ implies that
$H^S$ is $(m,E,r)$-regular.

We now consider the event  
\be
\cE_K&= \set{ I_k\setminus(\what Y_{{K^\pr,K^c}} \cup \what Y_{K^c})\ne \emp}= \set{ I_k \cap Y_{{K^\pr,K^c}}\cap Y_{K^c}\ne \emp}  \\
& \subset \set{ \dist \set{\sigma_{\what I_k}(H^{{{K^\pr,K^c}} }), \sigma_{\what I_k} (H^{K^c})}\le 2  e^{-mr}},
\ee
and note that it follows from Lemma \ref{Wegnerineq}(ii) that
\be
\P\set{\cE_K} \le C_k \abs{\Lambda}^{4k+1} e^{-mr}.
\ee

Since
\be
\P\set{ \cE_K \cup \cJ_{{K^\pr,K^c}} \cup \cJ_{K^c}} \le C_k \abs{\Lambda}^{4k+1} e^{-mr}+ 2C_k \abs{\La}^{\xi_k+1} e^{-2mr} \le  C\abs{\La}^{\xi^\pr_k} \e^{-m r},
\ee
and on the complementary event  we have  $I_k= \what Y_{{K^\pr,K^c}} \cup \what Y_{K^c}$, so for $E\in I_k$ either $H^{{K^\pr,K^c}}$ or $H^{K^c}$ is $(m,E,r)$-regular, the lemma is proved. 
\end{proof}

\begin{proof}[Proof of Corollary~\ref{thm:eigencor}]
Let  $A\subset B\subset \La$, $A$ connected in $\La$,  let $r=\rho^\La(A,B)$, and recall
$\norm{P_-^{A}f(H^\Lambda) P_+^{B}}\le \norm{P_-^{A}f(H^\Lambda) P_+^{[A]^\La_r}}$.

We set  
\be
\Theta^\La(A,r)& =\sup_{\substack{f\in B(I_{\le k}):\\\|f\|_\infty\le1}}\norm{P_-^{A}f(H^\Lambda) P_+^{[A]^\La_r}}\le 1.
\ee
To estimate $\E\set{ \Theta^\La (A,r)}$,  note that
\be\label{Thetasum}
 \Theta^{\La}(A,r) \le   \sum_{E\in \sigma_{I_k}(H^{\La})} \norm{{P_-^{A}} P_{\set{E}}P_+^{[A]^\La_r}}, \qtx{where}P_{\set{E}}= \chi_{\set{E}}(H^\La).
 \ee
The spectrum of $H^\La$ is simple almost surely, as commented in  \cite[Section 3]{EKS1},  so we assume this on what follows for simplicity.  (Otherwise we just  need to label the eigenvalues taking into account multiplicity.) For $E\in \sigma(H^\La)$ we let $\phi_E$ denote the corresponding eigenfunction, and let $N_E\in \N^0$ be given by $\cN_\La \phi_E= N_E \phi_E$.

For $E\in I_k$  we have 
\be
P_{\set{E}}&=\what  R^{\Lambda}_{k,E}\pa{\what  H^\Lambda_k-E}P_{\set{E}}=\what  R^{\Lambda}_{k,E}\pa{\what  H^\Lambda_k-H^\Lambda+(H^\Lambda -E)}P_{\set{E}}\\&=k\tfd\what  R^{\Lambda}_{k,E}\,Q_{\le k}^{\Lambda}P_{\set{E}}.
\ee

{  Let   $r \ge R_k =  6k(\cl{\frac {18} {\theta_k}}+2)$. }  Using \eq{eq:stidena} and \eq{CT}, we obtain
\be\label{secter5}
{\norm{{P_-^{A}}P_{\set{E}} P_+^{[A]^\La_r}}}&=k \tfd {\norm{{P_-^{A}}\what  R^{\Lambda}_{k,E}\,Q_{\le k}^{\Lambda}P_{\set{E}} P_+^{[A]^\La_r}}}\\
&{  =k \tfd {\norm{{P_-^{A}}\what  R^{\Lambda}_{k,E}P_-^{[A]_\infty^\La}\,Q_{\le k}^{\Lambda}P_{\set{E}} P_+^{[A]^\La_r}}}}
\\ &\le k { \sum_{q=-\abs{A}}^{\abs{\Lambda}} } {\norm{{P_-^{A}}\what  R^{\Lambda}_{k,E}P_+^{[A]^\La_q}}}{\norm{P_-^{]A[^\La_q}  Q_{\le k}^{\Lambda}P_{\set{E}}P_+^{[A]^\La_r}}} \\ 
&\le C_0{ \sum_{q=-\abs{A}}^{\abs{\Lambda}} }  \e^{-m_0(q)_+}{\norm{P_-^{]A[^\La_q}  Q_{\le k}^{\Lambda}P_{\set{E}}P_+^{[A]^\La_r}}}
 \\ 
&\le 2 C_0{ \sum_{q=-\abs{A}}^{r-1 - R_k}} \e^{-m_0(q)_+}  \sum_{u\in ]A[^\La_q}  {\norm{  Q_{\le k}^{\Lambda}\cN_uP_{\set{E}}P_+^{[u]^\La_{r-q-1}}}} + C_k\abs{\La} \e^{-m_0 r}.
\ee

{ Let     $u\in \La$ and $p\ge  R_k$.}     If $8kN_E \ge p$, it follows from   \eq{PBkB'}-\eq{PVP'HN} that  
 \be\label{p<8kN}
 &\norm{ \chi_{N_E}^\La  Q_{\le k}^{\Lambda}\cN_uP_{\set{E}}P_+^{[u]^\La_{p}}}\le \chi_{{{\mathcal B_k^{N_E}}}},\\
&\P_\La \pa{{\mathcal B_k^{N_E}}} 
\le C_k \abs{\Lambda}^{2k}e^{- c_\mu {N_E}}\le  C_k \abs{\Lambda}^{2k}\e^{- \frac {c_\mu}{8k}p}.
\ee

If $p> 8k{N_E}$, we set (cf.  \eq{eq:Kse})
\be\label{eq:Kse99} 
K(0)&= [u]^\La _{\frac {3p}4}   \qtx{and} K(a)=[ u]^\La_{a  \fl{\frac p {6k}}} \mqtx{for}a=1,2,\ldots, 3k-1,\\
S(a)& =[\partial^\La K(a)]^\La_{\fl{\frac p {6k}}-1}  \qtx{for} a=0,1,\ldots, 3k-1.
\ee
Using  Lemma \ref{lem:partM}, we get
 \be\label{dynloc1}
 {\Big\|{ \chi^\La_{N_E} Q_{\le k}^{\Lambda}\cN_uP_{\set{E}}P_+^{[u]^\La_{p}}}\Big\|}  \le 
   \sum_{a=0}^{3k-1}  {\Big\| \chi_{N_E}^\La Q_{\le k}^{\Lambda}\cN_u Y(a) P_{\set{E}}P_+^{[u]^\La_{p}}\Big\|},
\ee 
 where  $Y(0)=P_+^{  \La  \setminus  [u]^\La _{\frac p2} } $ and $Y(a)=P_+^{S(a)}P_-^{K(a)}P_-^{K^c(a)}$ for $a>0$.

 We now consider the event (see \eq{MSAdef})
 \be\label{Jup}
&  \cJ_k (u,p) =\bigcap_{a=0}^{3k-1} \cF_{k}^\La (K(a), {\what  {\theta_k}} ,\what  p),     \sqtx{where} \what  {\theta_k}=\tfrac {{\theta_k}} 9\sqtx{and}\what  p= \fl{\tfrac p {6k}}-1\ge  \tfrac {18} {{\theta_k}},
 \ee
 and note that it follows from Lemma \ref{lemMSAprob} that
 \be\label{PJupc}
 \P\set{ \pa{\cJ_k (u,p)}^c}\le 3k C\abs{\La}^{\xi^\pr_k} \e^{-\what  {\theta_k}\what  p}.
 \ee
For  $\omega \in  \cJ_k (u,p) $ and  $a\in \set{0,1,\ldots, 3k-1}$, either $H^{\pa{K(a), K^c(a)} }$  or $H^{{K^c}(a)}$ is $({\what  {\theta_k}} ,E,\what  p)$-regular (${K^c}(a)=\pa{K(a)}^c$). 
 If $H^{{K^c}(a)}$ is $({\what  {\theta_k}} ,E,\what  p)$-regular,
 we note that  
\be
P_{\set{E}}P_+^{[u]^\La_{p}} &=P_{\set{E}}\pa{H^{K(a)}+H^{{K^c}(a)}-E}  R^{{K^c}(a)}_EP_+^{[u]^\La_{p}}\\
&  =    -P_{\set{E}}\Gamma^{\pa{K(a),K^c(a)}}P_-^{\partial^\La_{ex} K(a)}P_+^{K(a)} R^{{K^c}(a)}_EP_+^{[u]^\La_{p}}  ,
\ee
where we have used $R^{{K^c}(a)}_EP_+^{[u]^\La_{p}}=P_+^{K(a)}R^{{K^c}(a)}_EP_+^{[u]^\La_{p}}$ due to $K(a)\subset [u]^\La_{p}$.
We deduce that
\be\label{dynloc2}
 {\norm{ \chi_{N_E}^\La Q_{\le k}^{\Lambda}\cN_u Y(a) P_{\set{E}}P_+^{[u]^\La_{p}}}}\le 
\norm{P_{\set{E}}P_+^{[u]^\La_{p}}}\le \tfrac 1 \Delta \norm{P_-^{\partial_{ex}^\La K(a)} R^{{K^c}(a)}_EP_+^{[u]^\La_{p}\cap ({K^c}(a))}}\le  \tfrac {2\e^{-\what  {\theta_k}\what  p}} \Delta ,
 \ee
using \eqref{hPN},  \eqref{defregular}, and the definition of $K(a)$.
If  $H^{\pa{K(a),K^c(a)}}$  is $({\what  {\theta_k}} ,E,\what  p)$-regular, we use
\be\label{eq:expG5}
\cN_{u}P_{\set{E}}P_+^{[u]^\La_{p}}  &=\cN_{u}R^{\pa{K(a),K^c(a)}}_E\pa{H^{\pa{K(a),K^c(a)}}-E}P_{\set{E}}P_+^{[u]^\La_{p}} \\ &=-\cN_{u}R^{\pa{K(a),K^c(a)}}_E P_-^{\partial^\La K(a)}\Gamma^{\pa{K(a),K^c(a)}}P_{\set{E}}P_+^{[u]^\La_{p}}.
\ee
Thus
\be\label{dynloc3}
{\norm{ \chi_{N_E}^\La Q_{\le k}^{\Lambda}\cN_u Y(a) P_{\set{E}}P_+^{[u]^\La_{p}}}}&\le \norm{\cN_{u} Y(a)P_{\set{E}}P_+^{[u]^\La_{p}}}\le \tfrac 1 \Delta \norm{\cN_{u} Y(a)R^{\pa{K(a),K^c(a)}}_E P_-^{\partial^\La K(a)}}\\
 &\le  \tfrac 1 \Delta \norm{ P_+^{S(a)}R^{\pa{K(a),K^c(a)}}_E P_-^{\partial^\La K(a)}}\le 
 \tfrac 2 \Delta  \e^{-\what  {\theta_k}\what  p},
\ee
using \eqref{hPN},  \eqref{defregular}, and   the definition of $S(a)$.

Combining \eq{dynloc1},  \eq{dynloc2} and  \eq{dynloc3},    we conclude that for  $p>8kN_E$ and $\omega \in  \cJ_k (u,p) $ we have
\be
 {\norm{ \chi^\La_{N_E} Q_{\le k}^{\Lambda}\cN_uP_{\set{E}}P_+^{[u]^\La_{p}}}} \le  \tfrac {12k} \Delta  \e^{-\what  {\theta_k}\what  p}.
\ee
Since  $ {\norm{ \chi^\La_{N_E} Q_{\le k}^{\Lambda}\cN_uP_{\set{E}}P_+^{[u]^\La_{p}}}} \le 1$, it follows that for 
 $p>8kN_E $ we have
 \be\label{p>8kN}
  {\norm{ \chi^\La_{N_E} Q_{\le k}^{\Lambda}\cN_uP_{\set{E}}P_+^{[u]^\La_{p}}}} \le  \tfrac {12k} \Delta  \e^{-\what  {\theta_k}\what  p} + \chi_{ {\cJ_k (u,p)}^c}.
 \ee

It follows that for  $u\in \La$  { and  $p\ge  R_k$}, using  \eq{p<8kN} , \eq{p>8kN}, and \eq{trkH},  we conclude that 
\be
\E \Big({ \sum_{E\in \sigma_{I_k}(H^{\La})}\Big\|{ Q_{\le k}^{\Lambda}\cN_uP_{\set{E}}P_+^{[u]^\La_{p}}}  }\Big\|\Big)
\le C_k \abs{\La}^{\xi^{\pr}_k}\e^{-  \theta^\pr_k  p}.
\ee
Combining with
\eq{Thetasum},  \eq{secter5}, \eq{trkH},  we obtain
\be\label{EThetark}
\E\set{ \Theta^\La (A,r)}\le C_k   \abs{\La}^{  \xi_k^\pr} \e^{-   {\theta_k^\pr} r}.
\ee
The estimate \eq{EThetark} holds for $r\ge R_k$.  Since $\E\set{ \Theta^\La (A,r)}\le 1$ for all $r\ge 0$,  it holds for all $r\ge 0$ if the constant $C_k$ is replaced by the constant  $\wtilde C_k= C_k \e^{  {\theta_k^\pr} \what R_k}$.
\end{proof}
\appendix

\section{Useful identities}\label{appA}

 In  this appendix we list some useful identities. Their derivations are straightforward,  so we leave out the proofs.

   We fix $\La \subset \Z$ finite,
   
  \begin{itemize}
\item 
 For  all
 $i,j\in \La$ we have (recall \eq{P+-S})
\be\label{P-ij}
P_-^{\set{i}}&=\cN_i, \\
P_-^{\set{i,j}}&=\cN_i +\cN_j -\cN_i \cN_j=  \cN_i\pa{1-\cN_j} + \cN_j= P_+\up{\set{j}}\cN_i+ \cN_j.
 \ee

\item 
Consider the self-adjoint operator ${h}_{i,i+1}$ (recall \eqref{tildeh}) on the four-dimensional Hilbert space  $\cH_{\set{i,i+1}}=\C_i^2\otimes \C_{i+1}^2$.
An explicit calculation shows  that ${h}_{i,i+1}$ has eigenvalues $-1,0, \pm \frac 1 \Delta$.
It follows that if
$\set{i,i+1}\subset \La $ we have
 \beq\label{nth}
\norm{{h}_{i,i+1}}=1   \qtx{on} \cH_\La.
\eeq

\item The following  identities hold on $\cH_\La$ for $\set{i,i+1}\subset \La$:
\be\label{hPN}
&{h}_{i,i+1}P_+^{\set{i,i+1}}=P_+^{\set{i,i+1}}{h}_{i,i+1}=0,\\ 
&\norm{P_+^{\set{i}}{h}_{i,i+1}}=\norm{P_+^{\set{i+1}}{h}_{i,i+1}}={ \tfrac 1{2\Delta}},\\ 
&P_+\up{\set{i}}{h}_{i,i+1}P_+\up{\set{i}}=P_+\up{\set{i+1}}{h}_{i,i+1}P_+\up{\set{i+1}}=0,\\ 
&{h}_{i,i+1} \cN_i\cN_{i+1} = \cN_i\cN_{i+1} {h}_{i,i+1}= \cN_i\cN_{i+1}{h}_{i,i+1} \cN_i\cN_{i+1} .
\ee
  In particular,   the first identity above  implies
\beq\label{hPN9}
{h}_{i,i+1}= {h}_{i,i+1}P_-^{\set{i,i+1}}=P_-^{\set{i,i+1}}{h}_{i,i+1}=P_-^{\set{i,i+1}}{h}_{i,i+1}P_-^{\set{i,i+1}}.
\eeq

\item  Let $K\subset \La$, and recall \eq{eq:Gamma}. It follows from \eqref{hPN9} that
\beq\label{hPN98}
 \Gamma^K= P_-^{ \partial^\La K} \Gamma^K P_-^{ \partial^\La K}.
\eeq
If $K$ is connected in $\La$, it follows from \eqref{hPN98} that 
\beq\label{PGamma1}
 \norm{P_+^{K} \Gamma^K}\le \tfrac  1 \Delta  \qtx{and}   \norm{P_+^{K^c} \Gamma^K}\le \tfrac  1 \Delta.
 \eeq

\item  The  following identities hold for any non-empty $M\subset\La$ (recall \eq{P+-empty}):
 \be\label{eq:stidena}
 P_-^{[M]_\infty} P_+^M  &  =
   \sum_{q=0}^{\abs{\Lambda}} P_+^{[M]^\La_{q}}P_-^{]M[^\La_q}= \sum_{q=0}^{\abs{\Lambda}} P_+^{[M]^\La_{q}} P_-^{{\partial_{ex}^\La{[M]_{q}}}} ,\\
   P_-^{ M}&= \sum_{q=-\abs{M}}^{-1}P_+^{[M]^\La_{q}} P_-^{]M[^\La_q}=\sum_{q=-\abs{M}}^{-1}P_+^{[M]^\La_{q}} P_-^{{\partial_{in }^\La{[M]_{q+1}}}}  ,
 \\P_-^{[M]^\La_\infty} &= \sum_{q=-\abs{M}}^{\abs{\La}}P_+^{[M]^\La_{q}} P_-^{]M[^\La_{q}}. 
\ee

\end{itemize}

\section{Many-body quasi-locality} \label{app:quasil}
In this appendix we prove \eqref{eq:locmany1}.  Recall we only consider finite subsets of $\Z$.
We fix $\La \subset \Z$ and consider the Hilbert space $\cH_\La$.

\begin{lemma} \label{lem:F}
Suppose that  $H\in\mathcal A_{\La}$  satisfies
\begin{enumerate}
\item  For all   $K \subset {\La} $   we have   $[P_-^{K},H]P_+^{[K]_1^{\La}}=0$.

 \item For all  connected  $K \subset {\La} $  we have  $\norm{[P_-^{K},H]}\le  \gamma  $. 

 \end{enumerate} 
Then for all $A\subset B\subset {\La}$,  $A$  connected in ${\La}$,   we have
 \be \label{eq:loca1}
\norm{P_-^{A}\,e^{itH}\,P_+^{B}}\le  \gamma^{r}\frac{  \abs{t}^r}{r!}, \qtx{where} r= \dist_{\La} \pa{A,B^c}\ge 1.
\ee 
 \end{lemma} 
\begin{proof}

We  note that $[ A]^{\La}_{s} \subset B$ for $s=0,1,\ldots,r-1$.   We have
\be 
P_-^{ A}\,e^{itH}\,P_+^{ B} =ie^{itH}\int_0^tK(s)\,P_+^{B}ds, 
\ee
where $K(s)=e^{-isH}\,[P_-^{  A},H]\,e^{isH}$.  If $r\ge 2$,   condition (i) of  the Lemma yields $K(s)=e^{-isH}\,[P_-^{  A},H] P_-^{ [A]^{\La}_1}\,e^{isH}$. Proceeding recursively, we get
\be
P_-^{ A}\,e^{itH}\,P_+^{ B}&=i^{r}\int_0^t\int_0^{s_1}\ldots\int_0^{s_{r-1}}\prod_{j=1}^{r} K_{j-1}(s_j)ds_j \,P_+^{B},\\ K_j(s)&=e^{-isH}\,[P_-^{  [A]_j},H] \,e^{isH}.
\ee
Using assumption  (ii), we get 
\be \label{PATt'}
\norm{P_-^{ A}\,e^{itH}\,P_+^{ B}}\le   \gamma^{r}\frac{  \abs{t}^r}{r!} .
\ee
\end{proof}

\begin{lemma}
Let $f\in C_0^n$, i.e., $f$ is compactly supported and $n$ times differentiable function on $\R$ (with $n\ge2$). Then for $A,B,H$  as in  Lemma~\ref{lem:F}  and $r=\dist_\La \pa{A,B^c}$, we have 
\be\label{eq:Fbn}
\norm{P_-^A\ f(H)\, P_+^B}\le \wtilde C(f,n)r^{-(n-1)\min(1,\frac r n)} \le \wtilde C(f,n)r^{-n}.
\ee
\end{lemma}
\begin{proof}
Let $\hat f$ denote the Fourier transform of $f$, then we have $\abs{\hat f(t)}\le C(f,n)\langle t\rangle^{-n}$ for $t\in\R$ (we recall that $\langle t\rangle:=\sqrt{1+t^2}$). We can bound
\be\label{eq:2int}
\norm{P_-^A\ f(H)\, P_+^B}\le \int_{\mathcal R}\norm{P_-^A\ e^{itH}\, P_+^B}\abs{\hat f(t)}dt+\int_{\mathcal R^c}\abs{\hat f(t)}dt,
\ee
where $\mathcal R:=[-R,R]$, where $R>0$ will be chosen later. 

We can bound the first integral on the right hand side of \eqref{eq:2int} using \eqref{eq:loca1} as
\be
\int_{\mathcal R}\norm{P_-^A\ e^{itH}\, P_+^B}\abs{\hat f(t)}dt  &  \le C(f,n)\frac{\gamma^{r}}{r!}\int_{\mathcal R} \abs{t}^r\langle t\rangle^{-n}dt\le  C_n C(f,n)\frac{\gamma^{r}{R}^{1+(r-n)_+}}{r!}\\
 & \le C'_n C(f,n)\pa{\frac {\e \gamma }r}^r{R}^{1+(r-n)_+},
\ee
where we used $r!\ge e^{1-r}r^r$.

 On the other hand, we can bound the second integral in \eqref{eq:2int} as 
\be
\int_{\mathcal R^c}\abs{\hat f(t)}dt\le C(f,n)\int_{\mathcal R^c} \langle t\rangle^{-n}dt\le  C_n C(f,n)\pa{1+R}^{1-n}\le C_n C(f,n){R}^{1-n}.
\ee
Choosing $R=\pa{\frac r{e\gamma}}^{\frac{r}{n+ (r-n)_+}}$ , we get  \eqref{eq:Fbn}.
\end{proof}

\printbibliography

\end{document}